\documentclass[12pt,a4paper]{article}
\usepackage[utf8]{inputenc}
\usepackage{amsmath}
\usepackage{amsfonts}
\usepackage{amssymb}
\usepackage{graphicx}
\usepackage[left=2.00cm, right=2.00cm, top=2.00cm, bottom=2.00cm]{geometry}
\usepackage{sectsty}
\usepackage{amsthm}
\usepackage[colorlinks,citecolor=red,linkcolor=blue]{hyperref}
\usepackage{units}
\usepackage{titling}
\title{Boltzmann weights and fusion procedure for the rational seven-vertex SOS model}

\date{}

\newtheoremstyle{break}
{\topsep}{\topsep}%
{}{}%
{\bfseries}{}%
{\newline}{}%
\theoremstyle{break}

\makeatletter
\renewenvironment{proof}[1][\proofname]{%
	\par\pushQED{\qed}\normalfont%
	\topsep6\p@\@plus6\p@\relax
	\trivlist\item[\hskip\labelsep\bfseries#1]%
	\ignorespaces
}{%
	\popQED\endtrivlist\@endpefalse
}
\makeatother

\theoremstyle{break}

\newtheorem{prop}{Proposition}

\begin{document}
	
	
	
	
	
	
	\vspace*{1cm}
	
	\begin{center}
		{\Large \bf{Boltzmann weights and fusion procedure for the rational seven-vertex SOS model}}
	\end{center}
	
	\vspace{1.5 cm}
	
	\begin{center}
		P. V. Antonenko${}^{\dagger\bullet}$\footnote{ \sc e-mail: antonenko\_pavel@pdmi.ras.ru}, P. A. Valinevich${}^\dagger$\footnote{ \sc e-mail: valinevich@pdmi.ras.ru}
	\end{center}

	\vspace{0.8 cm}
	
	{\noindent \it ${}^\dagger$St. Petersburg Department of Steklov Mathematical Institute, St. Petersburg, Russia \\
	${}^\bullet$Leonhard Euler International Mathematical Institute in Saint Petersburg, St. Petersburg, Russia}
	
	\vspace{1.5 cm}
	
	{\bf Abstract}
	
	We consider seven-vertex two-dimensional integrable statistical model. With the help of intertwining vector method we construct its counterpart integrable model of SOS type. More general models of both types are constructed by means of fusion procedure. For SOS models we calculate the Boltzmann weights in terms of terminating hypergeometric series ${}_{9} F_8.$ Then using the similarity transformation for $R$-operators we construct a new family of vertex models containing the $11$-vertex model as the simplest representative. For this new set of models the vertex-SOS correspondence is constructed: we find the intertwining vectors, show that they do not depend on spectral parameter and the SOS statistical weights are similar to those obtained from the $7$-vertex model.

	\vspace{1.5 cm}

		
		
		


	\section{Introduction}
	
	Among the known integrable models solved with the help of quantum inverse scattering method there is a large class of statistical models defined on a two-dimensional lattice \cite{baxter}. In this work we deal with two classes of such models: vertex models \cite{bogolubov,slavnov,saleur} and models of SOS (solid-on-solid) type \cite{baxter,date88,pasquier}. For the first class of models the state of the lattice is defined by the state of its edges, and each lattice node (vertex) is assigned a Boltzmann weight; for the second -- the state parameters are attached to the vertices of the lattice, and Boltzmann weights are attributed to faces bounded by edges. For both these classes the condition of integrability has the form of a nonlinear equation for statistical weights -- the Yang-Baxter equation (which in addition depends on imposed boundary conditions \cite{skl1}).
	
	At the dawn of quantum inverse scattering method R.~Baxter \cite{baxter73} introduced a method which allowed to construct for some integrable vertex models the corresponding integrable SOS models. First formulated for the eight-vertex model this technique was later generalized in the works \cite{date86,date88,konno,vega}. The crucial role is played by intertwining vectors. The Boltzmann weights of SOS models obtained by this method are matrix elements of the $R$-operator in the basis of intertwining vectors.
	
	In the present paper we apply the technique of intertwining vectors to the rational seven-vertex model which is a special limit of the eight-vertex model \cite{zabrodin}. The key element of the construction is the realization of $R$-matrix (the matrix of Boltzmann weights) in terms of an operator in the space of homogeneous polynomials. For the model under consideration it was obtained in \cite{VA_19}. A similar approach was used in \cite{der_ch} for the six-vertex model, and in our case it allowed to considerably simplify the calculations.
	
	The result of the work is the expression for the Boltzmann weight $W^{(n,m)}$ of the SOS model for arbitrary $n, m$ -- the numbers of lattice edges' states. It has the form of the hypergeometric series ${}_{9} F_8.$ Representations of this type were obtained earlier for other models, see e.g. \cite{date88}.
	
	The structure of the paper is as follows. In section~\ref{models_section} we give the definitions of vertex and SOS models, define the intertwining vectors which are essential to construct the correspondence between them. Section~\ref{fusion_section} contains the necessary information about the fusion procedure for vertex and SOS models. It is needed for construction of the most general models with arbitrary number of faces or vertices states. In section~\ref{7V_model_section} we define the rational seven-vertex model, find the corresponding intertwining vectors and SOS model, discuss the connection with the SOS model from \cite{zapiski}. In section~\ref{polynomial_sect} the representation of the $R$-operator in terms of the differential operator on the space of polynomials \cite{der_ch} is derived, it allows to find the Boltzmann weights $W^{(n,m)}$ in the case of arbitrary $n,m\in\mathbb{N}$. It is shown how the fusion procedure works in this representation. The expression for Boltzmann weights $W^{(n,m)}$ in the form of terminating hypergeometric series ${}_{9}F_8$ is obtained in section~\ref{Wnm_section} with the help of realization of the $R$-operator and intertwining vectors in the space of polynomials. In section~\ref{11V_section} from the considered family of vertex models we construct a new family in which the simplest representative is the $11$-vertex model \cite{AZ_23,DKK_03,KS_97}, and also obtain the vertex-SOS transform for these models.
	
	\medskip
	\medskip
	\medskip
	
	\textbf{Acknowledgements}
	
	The study by P. V. Antonenko was conducted at Leonhard Euler International Mathematical Institute in Saint Petersburg with the financial support from The Ministry of Science and Higher Education of the Russian Federation (project no. 075-15-2022-289).

	\section{Vertex and SOS models} \label{models_section}

	\subsection{Vertex models}
	
	\begin{figure}[t]
		\begin{minipage}{0.45\textwidth}
			\centering\includegraphics[scale=0.5]{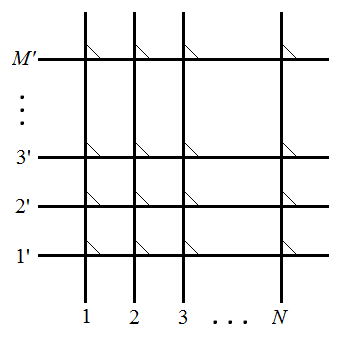}
			\caption{lattice of a vertex model.}
			\label{vertex_lattice}
		\end{minipage}
		\quad
		\begin{minipage}{0.45\textwidth}
			\centering\includegraphics[scale=0.5]{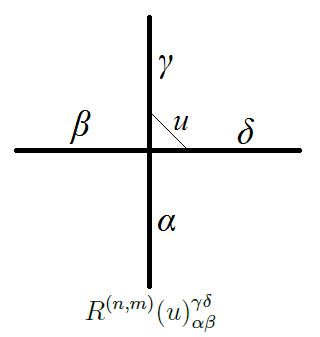}
			\caption{Boltzmann weight of a vertex.}
			\label{vertex_statistical_weight}
		\end{minipage}
	\end{figure}

	We consider vertex models defined on a two-dimensional square lattice, see fig.~\ref{vertex_lattice}, numbers of horizontal rows are written with prime. To its edges one assigns the state parameters. A set of state parameters of all edges is called a state of the lattice. In the models under our consideration the state parameters on horizontal lines take values from the set $\{-m+2j,\:j=0,1,\ldots,m\}$ ($m+1$ elements), and on vertical lines -- from the set $\{-n+2j,\:j=0,1,\ldots,n\}$ ($n+1$ elements). A vertex with parameters $\alpha,\beta,\gamma,\delta$ on adjacent edges (fig.~\ref{vertex_statistical_weight}) is assigned a Boltzmann weight $R^{(n,m)}(u)_{\alpha\beta}^{\gamma\delta}$, which is a function of the spectral parameter $u.$ The quantity $u$ reflects the dependence of vertex energies and Boltzmann weights on external parameters~\cite{baxter}. The energy of the lattice is the sum of energies of all vertices, then the statistical weight of a state is the product of Boltzmann weights of all vertices in the given state. The partition function of the lattice $Z^{(n,m)}_{NM}(u)$ is equal to the sum of statistical weights of all system states.
	
	The quantity $R^{(n,m)}(u)_{\alpha\beta}^{\gamma\delta}$ can be treated as a matrix element of the operator $R^{(n,m)}(u)$ in the space $V^n\otimes V^m$ called the $R$-operator \cite{baxter,bogolubov,slavnov}:
	\begin{equation} \nonumber
		R^{(n,m)}(u) \, e^\gamma\otimes h^\delta=\sum_{\alpha,\beta} R^{(n,m)}(u)_{\alpha\beta}^{\gamma\delta} \, e^\alpha\otimes h^\beta,
	\end{equation}
	where $V^n\cong\mathbb{C}^{n+1}$ and $V^m\cong\mathbb{C}^{m+1}$ are vector spaces over $\mathbb{C}$, $\{e^{\alpha^\prime},\:\alpha^\prime=-n,-n+2,\ldots, n\}$ is a basis in $V^n$, $\{h^{\beta^\prime},\:\beta^\prime=-m,-m+2,\ldots, m\}$ is a basis in $V^m$.
	
	To $i$-th vertical line in fig.~\ref{vertex_lattice} we associate a vector space $V^n_i\cong\mathbb{C}^{n+1}$ over $\mathbb{C}$, and to $j$-th horizontal line -- the space $V^m_{j^\prime}\cong\mathbb{C}^{m+1}$ \cite{slavnov}. To the intersection of $j$-th horizontal and $i$-th vertical line one associates the tensor product $V^n_i\otimes V^m_{j^\prime}$ and the $R$-operator $R^{(n,m)}_{ij^\prime}(u)\in\mathrm{End}\left(V^n_i\otimes V^m_{j^\prime}\right)$.
	
	We will focus on the case of periodic boundary conditions. In this instance the partition function can be calculated exactly if the Yang-Baxter equation holds \cite{bogolubov,slavnov}:
	\begin{equation} \label{YBE_vert}
		R^{(k,n)}_{12}(v)R^{(k,l)}_{13}(u)R^{(n,l)}_{23}(u-v)=R^{(n,l)}_{23}(u-v)R^{(k,l)}_{13}(u)R^{(k,n)}_{12}(v) \, .
	\end{equation}
	If an $R$-matrix satisfies the equation (\ref{YBE_vert}), then the corresponding vertex model is called exactly solvable \cite{zapiski}.
	
	With the help of the fusion procedure \cite{krs81,ks_qism,date86,date88,konno}, which will be discussed later, it is possible to construct the matrices $R^{(n,m)}(u)$ of vertex models with arbitrary finite numbers of edge state parameters' values from the matrix $R^{(1,1)}(u)$.

	\subsection{SOS models}
	
	SOS (solid-on-solid) models are defined on square lattice consisting of $N$ vertical rows of faces and $M$ horizontal rows (fig.~\ref{IRF_lattice}). The state parameters are attached to the cells' vertices. The state of the lattice is defined as a set of state parameters of all vertices. We considered the models with integer-valued state parameters. A face with the parameters of vertices equal to $a,b,c,d$ is assigned a Boltzmann weight (fig.~\ref{IRF_statistical_weight})
	\begin{equation} \nonumber
		W^{(n,m)}\left(\left.\begin{array}{cc}
			a & b\\ d & c
		\end{array}\right|u\right)
	\end{equation}
	which depends on the spectral parameter $u$. In our models the statistical weights satisfy the following condition:
	\begin{equation} \label{IRF_model_condition}
		W^{(n,m)}\left(\left.\begin{array}{cc}
			a & b\\ d & c
		\end{array}\right|u\right) \neq 0 \Leftrightarrow
		\begin{cases}
			a-b, d-c\in \{-n+2j,\:j=0,1,\ldots,n\} \\
			a-d, b-c\in \{-m+2j,\:j=0,1,\ldots,m\}
		\end{cases},
	\end{equation}
	i.e. the absolute value of the difference between the parameters of adjacent vertices on a horizontal line is less or equal than $n,$ and on a vertical line -- less or equal than $m.$
	
	\begin{figure}[t]
		\begin{minipage}{0.45\textwidth}
			\centering\includegraphics[scale=0.5]{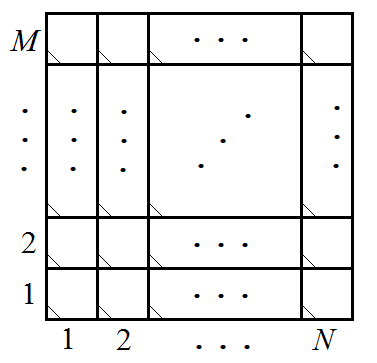}
			\caption{lattice of an SOS model.}
			\label{IRF_lattice}
		\end{minipage}
		\quad
		\begin{minipage}{0.45\textwidth}
			\centering\includegraphics[scale=0.4]{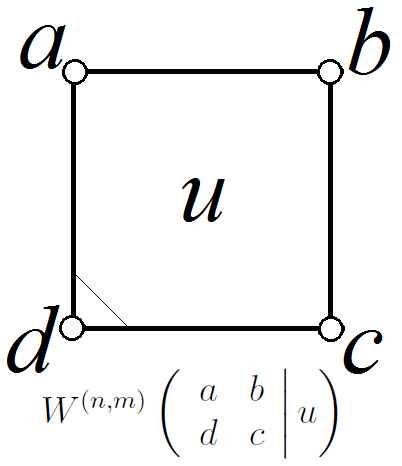}
			\caption{Boltzmann weight of a face.}
			\label{IRF_statistical_weight}
		\end{minipage}
	\end{figure}

	We imposed the periodic boundary conditions: on each vertical and horizontal line the state parameters of boundary vertices are equal. The model is exactly solvable if the SOS Yang-Baxter equation holds~\cite{baxter}
	\begin{align}
		\nonumber & \sum\limits_{g}
		W^{(k,n)}\left(\left.\begin{array}{cc}
			f & g\\
			e & d
		\end{array}\right|v-w\right)
		W^{(k,l)}\left(\left.\begin{array}{cc}
			a & b\\
			f & g
		\end{array}\right|u-w\right)
		W^{(n,l)}\left(\left.\begin{array}{cc}
			b & c\\
			g & d
		\end{array}\right|u-v\right) \\
		\label{YBE_IRF} & = \sum\limits_{g}
		W^{(n,l)}\left(\left.\begin{array}{cc}
			a & g\\
			f & e
		\end{array}\right|u-v\right)
		W^{(k,l)}\left(\left.\begin{array}{cc}
			g & c\\
			e & d
		\end{array}\right|u-w\right)
		W^{(k,n)}\left(\left.\begin{array}{cc}
			a & b\\
			g & c
		\end{array}\right|v-w\right).
	\end{align}
	Note that the sums are finite by virtue of the condition (\ref{IRF_model_condition}).

	\subsection{The correspondence between vertex and SOS models}
	
	The correspondence between $R^{(n,m)}$ and $W^{(n,m)}$ is given by the the expression \cite{konno}
	\begin{equation} \label{vertex_IRF}
		R^{(n,m)}(u-v) \, \psi^{(n)}(u)^a_b\otimes\psi^{(m)}(v)^b_c = \sum\limits_{b^\prime} \psi^{(n)}(u)^{b^\prime}_c\otimes\psi^{(m)}(v)^a_{b^\prime} \, W^{(n,m)}\left(\left.\begin{array}{cc}
			a & b\\
			b^\prime & c
		\end{array}\right|u-v\right),
	\end{equation}
	where $\psi^{(n)}(u)^a_b$ is a vector in $\mathbb{C}^{n+1}$, which depends on the spectral parameter $u$ and integer-valued state parameters of the SOS model $a,b$. It is called an intertwining vector. The equation \eqref{vertex_IRF} is called the vertex-SOS correspondence. By definition, $\psi^{(n)}(u)^a_b$ is non-zero if and only if $(a-b)\in\{-n+2j, \; j=0,1,2,\ldots,n\}$. This condition correlates with the restriction (\ref{IRF_model_condition}) on the SOS Boltzmann weights.
	
	The relation \eqref{vertex_IRF} was introduced by R.~Baxter for the eight-vertex model \cite{baxter73} (its vertex state parameters take $2$ values, i.e. $n=m=1$) and was generalized later to arbitrary $n$ and $m$ in \cite{date86,date88}.

	\section{The fusion procedure} \label{fusion_section}
	
	In this section we show how from the $R$-operator $R^{(1,1)}(u)$ describing the vertex model with two parameters of edges one can construct the $R$-operator $R^{(n,m)}(u)$ of the model with arbitrary number of edges states, it acts in $\mathbb{C}^{n+1}\otimes\mathbb{C}^{m+1}$. In order to construct this space from the spaces $\mathbb{C}^2$, in which $R^{(1,1)}(u)$ acts, one needs to consider the tensor product of $n+m$ copies of $\mathbb{C}^2$. Denote the first $n$ spaces $V_1, \ldots, V_n$, and the last $m$ spaces $V_{\bar{1}}, \ldots, V_{\bar{m}}$. Consider the symmetric components in the first $n$ and in the last $m$ copies, then the space $V^{(n,m)} = S^n\mathbb{C}^2 \otimes S^m\mathbb{C}^2$ is isomorphic to $\mathbb{C}^{n+1}\otimes\mathbb{C}^{m+1}$, because $S^k\mathbb{C}^2 \simeq \mathbb{C}^{k+1}$.
	
	Now define how $R^{(n,m)}(u)$ acts on the elements of $V^{(n,m)}$. Consider the operator
	\begin{equation} \label{R_k_1}
		R^{(n,1)}_{1\ldots n,\bar{j}}(u) = \Pi_{1\ldots n}R^{(1,1)}_{1\bar{j}}(u+n-1)\ldots R^{(1,1)}_{n-1\,\bar{j}}(u+1)R^{(1,1)}_{n\bar{j}}(u),
	\end{equation}
	where $\Pi_{1\ldots n}$ is the projector onto the space $S^n\mathbb{C}^2$, it acts on the basis elements in the following way:
	\begin{equation} \nonumber
		\Pi_{1\ldots n} \, e^{i_1} \otimes \ldots \otimes e^{i_n} = \frac{1}{n!}\sum\limits_{\sigma\in S_n} e^{i_{\sigma(1)}} \otimes \ldots \otimes e^{i_{\sigma(n)}},
	\end{equation}
	where $S_{n}$ is the set of permutations of the first $n$ natural numbers, and $i_1,\ldots,i_n$ take the values $+1$ and $-1$.
	The operator \eqref{R_k_1} acts nontrivially in the spaces $V_1, \ldots, V_n$ and $V_{\bar{j}}$. Now define $R^{(n,m)}(u)$:
	\begin{equation} \label{R_k_n}
		R^{(n,m)}(u) = \Pi_{\bar{1}\ldots\bar{m}}R^{(n,1)}_{1\ldots n,\bar{m}}(u)R^{(n,1)}_{1\ldots n,\overline{m-1}}(u-1)\ldots R^{(n,1)}_{1\ldots n,\bar{1}}(u-m+1).
	\end{equation}
	The operator $R^{(n,m)}(u)$ acts in the tersor product of all $n+m$ copies. But it can be seen from \eqref{R_k_1} and \eqref{R_k_n} that the result of its action belongs to $V^{(n,m)}$. The restriction of $R^{(n,m)}(u)$ to $V^{(n,m)}$ gives us the desired $R$-operator of the model with arbitrary number of edges' states.

	\medskip

	\begin{prop}[\cite{krs81, ks_qism, konno}] \label{YBE_prop}
		Let the operator $R^{(1,1)}(u)$ satisfy the Yang-Baxter equation
		\begin{equation} \nonumber
			R^{(1,1)}_{12}(u-v)R^{(1,1)}_{13}(u)R^{(1,1)}_{23}(v)=R^{(1,1)}_{23}(v)R^{(1,1)}_{13}(u)R^{(1,1)}_{12}(u-v) \, ,
		\end{equation}
		and let $R^{(1,1)}(-1) = c(I - P)$, where $c \neq 0$ is an arbitrary complex number, $I$ is the identity operator in $\mathbb{C}^2\otimes\mathbb{C}^2$, $P$ is the permutation operator: $P \, v\otimes w = w \otimes v$.
		
		Then any triple of operators $R^{(k,n)}(u)$, $R^{(k,l)}(u)$ and $R^{(n,l)}(u)$ constructed from $R^{(1,1)}(u)$ by means of the formula \eqref{R_k_n} satisfies the Yang-Baxter equation \eqref{YBE_vert}.
	\end{prop}

	\medskip

	Assume that for the $R$-matrix $R^{(1,1)}(u)$ we managed to find the intertwining vectors $\psi^{(1)}(u)^a_b$ and construct the corresponding SOS model with Boltzmann weights $W^{(1,1)}$ as in the formula (\ref{vertex_IRF}). Then it is possible to construct the SOS model with weights $W^{(n, m)}$ satisfying (\ref{IRF_model_condition}) for the $R$-matrix $R^{(n,m)}(u)$.
	
	Consider the intertwining vector $\psi^{(1)}(u)^a_b$ in $\mathbb{C}^2$ depending on the spectral parameter $u$ and integer-valued SOS state parameters $a$ and $b$. By definition, it obeys the condition
	\begin{equation} \label{vector_psi_1_condition}
		\psi^{(1)}(u)^a_b\neq 0\Leftrightarrow |a-b|=1
	\end{equation}
	and connect $R^{(1,1)}$ with $W^{(1,1)}$ by means of the vertex-SOS correspondence
	\begin{equation}\label{v_IRF_11} 
		R^{(1,1)}(u-v) \, \psi^{(1)}(u)^a_b\otimes\psi^{(1)}(v)^b_c = \sum\limits_{b^\prime}\psi^{(1)}(u)^{b'}_c\otimes\psi^{(1)}(v)^a_{b'} \,
		W^{(1,1)}\left(\left.\begin{array}{cc}
			a & b\\
			b' & c
		\end{array}\right|u-v\right).
	\end{equation}
	Define the intertwining vector
	\begin{equation}  \label{fused_vector}
		\psi^{(k)}(u)^a_b = \Pi_{1\ldots k} \, \psi^{(1)}(u+k-1)^a_{c_1}\otimes\ldots\otimes\psi^{(1)}(u+1)^{c_{k-2}}_{c_{k-1}}\otimes\psi^{(1)}(u)^{c_{k-1}}_b \,,
	\end{equation}
	where $c_1, \ldots, c_{k-1}$ satisfy the condition $|c_1-a| = |c_2-c_1| = \ldots = |b - c_{k-1}| = 1.$ From \eqref{vector_psi_1_condition} it follows that
	\begin{equation} \nonumber
		\psi^{(k)}(u)^a_b\neq 0\Leftrightarrow (b-a)\in\{-k+2j,j=0,1,2,\ldots,k\}.
	\end{equation}

	\begin{prop}[\cite{date86,date88}] \label{vertex_SOS_fusion_prop}
		
		Let $R^{(1,1)},$ $W^{(1,1)}$ and $\psi^{(1)}$ obey (\ref{v_IRF_11}), and let the definition \eqref{fused_vector} of the intertwining vector $\psi^{(k)}(u)^a_b$ be independent of $c_1, \ldots, c_{k-1}$ for $k=n,m$. Then the equation (\ref{vertex_IRF}) holds for $R^{(n,m)}$ obtained with the help of the fusion procedure (\ref{R_k_n}), and the SOS Boltzmann weights are given by the formula
		\begin{equation} \nonumber
			W^{(n,m)}\left(\left.\begin{array}{cc}
				a & b\\
				b^\prime & c
			\end{array}\right|u\right)
			=
			\sum\limits_{a_1,\ldots,a_{m-1}}\prod\limits_{i=1}^m W^{(n,1)}\left(\left.\begin{array}{cc}
				a_{i-1} & b_{i-1}\\
				a_i & b_i
			\end{array}\right|u-m+i\right) .
		\end{equation}
		Here $b_0=b, \, b_m=c, \, a_0=a, \, a_m=b^\prime$, the numbers $b_1, b_2, \ldots, b_{m-1}$ satisfy the condition
		\begin{equation} \nonumber
			|b_1-b| = |b_2-b_1| = \ldots = |c-b_{m-1}| = 1 \, ,
		\end{equation}
		and $W^{(n,1)}$ is defined as follows:
		\begin{equation}  \nonumber
			W^{(n,1)}\left(\left.\begin{array}{cc}
				\tilde{a} & \tilde{b}\\
				\tilde{b}^\prime & \tilde{c}
			\end{array}\right|u\right)
			=
			\sum\limits_{\tilde{b}_1,\ldots,\tilde{b}_{n-1}}\prod\limits_{i=1}^n
			W^{(1,1)}\left(\left.\begin{array}{cc}
				\tilde{a}_{i-1} & \tilde{a}_i\\
				\tilde{b}_{i-1} & \tilde{b}_i
			\end{array}\right|u+n-i\right)
		\end{equation}
		where $\tilde{a}_0 = \tilde{a}, \, \tilde{a}_n = \tilde{b}, \, \tilde{b}_0 = \tilde{b}^\prime, \, \tilde{b}_n = \tilde{c}$, and $\tilde{a}_1, \tilde{a}_2, \ldots, \tilde{a}_{n-1}$ meet the condition
		\begin{equation} \nonumber
			|\tilde{a}_1-\tilde{a}| = |\tilde{a}_2-\tilde{a}_1| = \ldots = |\tilde{b}-\tilde{a}_{n-1}| = 1 \, .
		\end{equation}
		In so doing, the definition of $W^{(n,1)}$ does not depend on the choice of $\tilde{a}_1, \tilde{a}_2, \ldots, \tilde{a}_{n-1}$, and the definition of $W^{(n,m)}$ is independent of how one chooses $b_1, b_2, \ldots, b_{m-1}$.
	\end{prop}

	\medskip

	\begin{prop} \label{YBE_SOS_from_vertex_prop}
		Consider a family of operators $R^{(n,m)}(u)$ such that \eqref{YBE_vert} holds for any triple $R^{(k,n)}(u)$, $R^{(k,l)}(u)$ and $R^{(n,l)}(u)$, and every operator $R^{(n,m)}(u)$ in this family is associated to an SOS model with Boltzmann weights $W^{(n,m)}$ by means of intertwining vectors $\psi^{(n)}(u)^a_b$, $\psi^{(m)}(u)^a_b$ and the relation (\ref{vertex_IRF}). Consider a set of three natural numbers $k, n, l$. Let the sets of intertwining vectors
		\begin{equation} \label{intertwining_vect_sets}
			\{\psi^{(i)}(u)^a_b\}_{b\in\{a-i+2j,\:j=0,1,\ldots,i\}},\quad\{\psi^{(i)}(u)^b_c\}_{b\in\{c-i+2j,\:j=0,1,\ldots,i\}}, \qquad i=k, n, l
		\end{equation}
		be lineraly independent for any $a,c\in\mathbb{Z}$. Then the triple of functions $W^{(k,n)}$, $W^{(k,l)}$ and $W^{(n,l)}$ obey the SOS Yang-Baxter equation \eqref{YBE_IRF}.
	\end{prop}
	\begin{proof}
		$ $\\
		The equation \eqref{YBE_vert} is equivalent to
		\begin{equation} \label{YBE_vert_3}
			R^{(k,n)}(u-v)R^{(k,l)}(u-w)R^{(n,l)}(v-w)=R^{(n,l)}(v-w)R^{(k,l)}(u-w)R^{(k,n)}(u-v) \, .
		\end{equation}
		Act on $\psi^{(k)}(u)^a_b\otimes\psi^{(n)}(v)^b_c\otimes\psi^{(l)}(w)^c_d$ by the LHS of \eqref{YBE_vert_3}. Using the vertex-SOS correspondence \eqref{vertex_IRF} three times one obtains
		\begin{align} 
			\nonumber & R^{(k,n)}(u-v) R^{(k,l)}(u-w) R^{(n,l)}(v-w) \, \psi^{(k)}(u)^a_b\otimes\psi^{(n)}(v)^b_c\otimes\psi^{(l)}(w)^c_d \\
			\nonumber & = \sum\limits_{b^\prime,c^{\prime\prime}}\psi^{(k)}(u)^{c^{\prime\prime}}_{d}\otimes\psi^{(n)}(v)^{b^\prime}_{c^{\prime\prime}}\otimes\psi^{(l)}(w)^a_{b^\prime} \\
			\label{YBE_vert_IRF_left} & \times \sum\limits_{c^\prime}
			W^{(k,n)}\left(\left.\begin{array}{cc}
				b^\prime & c^\prime\\
				c^{\prime\prime} & d
			\end{array}\right|u-v\right)
			W^{(k,l)}\left(\left.\begin{array}{cc}
				a & b\\
				b^\prime & c^\prime
			\end{array}\right|u-w\right)
			W^{(n,l)}\left(\left.\begin{array}{cc}
				b & c\\
				c^\prime & d
			\end{array}\right|v-w\right) .
		\end{align}
		Now act on $\psi^{(k)}(u)^a_b\otimes\psi^{(n)}(v)^b_c\otimes\psi^{(l)}(w)^c_d$ by the RHS of \eqref{YBE_vert_3}. Applying the vertex-SOS correspondence \eqref{vertex_IRF} three times we find that
		\begin{align}
			\nonumber & R^{(n,l)}(v-w)R^{(k,l)}(u-w)R^{(k,n)}(u-v) \, \psi^{(k)}(u)^a_b\otimes\psi^{(n)}(v)^b_c\otimes\psi^{(l)}(w)^c_d \\
			\nonumber & =\sum\limits_{b^{\prime},c^{\prime\prime}}\psi^{(k)}(u)^{c^{\prime\prime}}_d\otimes\psi^{(n)}(v)^{b^{\prime}}_{c^{\prime\prime}}\otimes\psi^{(l)}(w)^{a}_{b^{\prime}} \\
			\label{YBE_vert_IRF_right_2} & \times\sum\limits_{c^\prime}W^{(n,l)}\left(\left.\begin{array}{cc}
				a & c^\prime\\
				b^{\prime} & c^{\prime\prime}
			\end{array}\right|v-w\right)W^{(k,l)}\left(\left.\begin{array}{cc}
				c^\prime & c\\
				c^{\prime\prime} & d
			\end{array}\right|u-w\right)W^{(k,n)}\left(\left.\begin{array}{cc}
				a & b\\
				c^\prime & c
			\end{array}\right|u-v\right) .
		\end{align}
		Since the sets of the form (\ref{intertwining_vect_sets}) are linearly independent, the set
		\begin{equation} \label{intertwining_vect_tens_prods}
			\{\psi^{(k)}(u)^{c^{\prime\prime}}_d\otimes\psi^{(n)}(v)^{b^{\prime}}_{c^{\prime\prime}}\otimes\psi^{(l)}(w)^{a}_{b^{\prime}}:\:\psi^{(n)}(v)^{b^{\prime}}_{c^{\prime\prime}}\neq 0\},
		\end{equation}
		where $c^{\prime\prime}\in\{d-k+2j,\: j=0,1,\ldots,k\},\: b^\prime\in\{a-l+2j,\: j=0,1,\ldots,l\}$, is also linearly independent.
		Therefore, comparing in (\ref{YBE_vert_IRF_left}) and (\ref{YBE_vert_IRF_right_2}) the coefficients on the terms of the form (\ref{intertwining_vect_tens_prods}) one obtains that the relation
		\begin{align}
			\nonumber & \sum\limits_{c^\prime}
			W^{(k,n)}\left(\left.\begin{array}{cc}
				b^\prime & c^\prime\\
				c^{\prime\prime} & d
			\end{array}\right|u-v\right)
			W^{(k,l)}\left(\left.\begin{array}{cc}
				a & b\\
				b^\prime & c^\prime
			\end{array}\right|u-w\right)
			W^{(n,l)}\left(\left.\begin{array}{cc}
				b & c\\
				c^\prime & d
			\end{array}\right|v-w\right) \\
			\label{YBE_IRF_2} &
			= \sum\limits_{c^\prime}
			W^{(n,l)}\left(\left.\begin{array}{cc}
				a & c^\prime\\
				b^{\prime} & c^{\prime\prime}
			\end{array}\right|v-w\right)
			W^{(k,l)}\left(\left.\begin{array}{cc}
				c^\prime & c\\
				c^{\prime\prime} & d
			\end{array}\right|u-w\right)
			W^{(k,n)}\left(\left.\begin{array}{cc}
				a & b\\
				c^\prime & c
			\end{array}\right|u-v\right) .
		\end{align}
		holds for any $a,b,c,d,c^{\prime\prime},b^\prime$. Swapping around $u-v$ and $v-w$ in \eqref{YBE_IRF_2} and changing $c^{\prime\prime}\to e, \, b^\prime\to f, \, c^\prime\to g$ we obtain the equivalent relation \eqref{YBE_IRF}.
	\end{proof}

	\section{The rational seven-vertex model and the corresponding SOS model} \label{7V_model_section}
	
	In the present paper we consider the rational seven-vertex model, its $R$-matrix has the form
	\begin{equation} \label{R_7V_rat}
		R^{(1,1)}(u)=\begin{pmatrix}
			u+1 & 0 & 0 & 0\\
			0 & u & 1 & 0\\
			0 & 1 & u & 0\\
			\alpha^2 u(u+1) & 0 & 0 & u+1
		\end{pmatrix} .
	\end{equation}
	It is a limiting case of the trigonometric seven-vertex model \cite{zabrodin}.
	
	Let us find the intertwining vectors and the SOS Boltzmann weights $W^{(1,1)}$ corresponding to $R^{(1,1)}$. To do this multiply the both sides of the relation \eqref{vertex_IRF} with $n = m = 1$ by the permutation operator $P:$
	\begin{equation} \label{vertex_IRF_fundamental}
		\check{R}^{(1,1)}(u-v) \, \psi^{(1)}(u)^a_b\otimes\psi^{(1)}(v)^b_c = \sum\limits_{b^\prime} \psi^{(1)}(v)^a_{b^\prime}\otimes\psi^{(1)}(u)^{b^\prime}_c \, W^{(1,1)}\left(\left.\begin{array}{cc}
			a & b\\
			b^\prime & c
		\end{array}\right|u-v\right),
	\end{equation}
	where $\check{R}^{(1,1)}(u-v) = P R^{(1,1)}(u-v).$
	Among the vectors $\psi^{(1)}(u)^a_b$ only $X_l(u) \equiv \psi^{(1)}(u)^l_{l+1}$ and $Y_l(u) \equiv \psi^{(1)}(u)^l_{l-1}$ are non-zero.
	
	In order to find $X_l(u)$ choose $a=l$, $b=l+1$, $c=l+2$ in (\ref{vertex_IRF_fundamental}). Then the equation (\ref{vertex_IRF_fundamental}) takes the form
	\begin{equation} \label{vertex_IRF_Xl_2}
		\check{R}^{(1,1)}(u-v) \, X_l(u)\otimes X_{l+1}(v) = \lambda_l(u-v) \, X_l(v)\otimes X_{l+1}(u),
	\end{equation}
	where
	\begin{equation} \nonumber
		\lambda_l(u-v) =
		W^{(1,1)}\left(\left.\begin{array}{cc}
			l & l+1\\
			l+1 & l+2
		\end{array}\right|u-v\right).
	\end{equation}
	By virtue of the condition \eqref{vector_psi_1_condition} only one term remains in the RHS of \eqref{vertex_IRF_Xl_2}. Representing the vectors in the coordinate form
	\begin{equation} \nonumber
		X_l(u) =
		\begin{pmatrix}
			\beta_l(u) \\ \gamma_l(u)
		\end{pmatrix}
	\end{equation}
	one obtains the system of equations equivalent to (\ref{vertex_IRF_Xl_2}):
	\begin{align}
		\label{vertex_IRF_Xl_3_1} &
		(u-v+1) \, \beta_l(u) \, \beta_{l+1}(v) = \lambda_l(u-v) \, \beta_l(v) \, \beta_{l+1}(u) \, \\
		\label{vertex_IRF_Xl_3_2} &
		\beta_l(u) \, \gamma_{l+1}(v) + (u-v) \, \gamma_l(u) \, \beta_{l+1}(v) = \lambda_l(u-v) \, \beta_l(v) \, \gamma_{l+1}(u) \, , \\
		\label{vertex_IRF_Xl_3_3} &
		(u-v) \, \beta_l(u) \, \gamma_{l+1}(v) + \gamma_l(u) \, \beta_{l+1}(v) = \lambda_l(u-v) \, \gamma_l(v) \, \beta_{l+1}(u) \, , \\
		\label{vertex_IRF_Xl_3_4} &
		\alpha^2(u-v) \, (u-v+1) \, \beta_l(u) \, \beta_{l+1}(v) + (u-v+1) \, \gamma_l(u) \, \gamma_{l+1}(v) = \lambda_l(u-v) \, \gamma_l(v) \, \gamma_{l+1}(u) \, .
	\end{align}
	The product of RHS of \eqref{vertex_IRF_Xl_3_1} and \eqref{vertex_IRF_Xl_3_4} is equal to the product of RHS of \eqref{vertex_IRF_Xl_3_2} and \eqref{vertex_IRF_Xl_3_3}. Then the same must be true for the LHS, that is
	\begin{align}
		\nonumber & [\beta_l(u)\gamma_{l+1}(v) + (u-v)\gamma_l(u)\beta_{l+1}(v)][(u-v)\beta_l(u)\gamma_{l+1}(v) + \gamma_l(u)\beta_l(v)] \\
		\label{requirement_vert_IRF} & = (u-v+1)\beta_l(u)\beta_{l+1}(v) [\alpha^2(u-v)(u-v+1)\beta_l(u)\beta_{l+1}(v)+(u-v+1)\gamma_l(u)\gamma_{l+1}(v)].
	\end{align}
	This is the necessary condition for solvability of the system (\ref{vertex_IRF_Xl_3_1}~--~\ref{vertex_IRF_Xl_3_4}). Dividing \eqref{requirement_vert_IRF} by $(\gamma_l(u)\beta_{l+1}(v))^2$ we obtain the equation
	\begin{equation} \label{vertex_IRF_Xl_5}
		[z(l,u,v)-1]^2=\alpha^2(u-v+1)^2y(l,u)^2,
	\end{equation}
	where
	\begin{equation} \nonumber
		z(l,u,v) = \frac{\beta_l(u)\gamma_{l+1}(v)}{\gamma_l(u)\beta_{l+1}(v)}, \qquad y(l,u) = \frac{\beta_l(u)}{\gamma_l(u)}.
	\end{equation}
	Consider the solution of \eqref{vertex_IRF_Xl_5} of the form $z(l,u,v)-1=-\alpha(u-v+1)y(l,u)$. Substituting it into the previous equality one finds
	\begin{equation} \label{vertex_IRF_Xl_7}
		\frac{\gamma_{l+1}(v)}{\beta_{l+1}(v)}+\alpha(-v+1)=\frac{\gamma_l(u)}{\beta_l(u)}-\alpha u.
	\end{equation}
	Let $\beta_l(u)=1$. The LHS of (\ref{vertex_IRF_Xl_7}) does not depend on $u$, and the RHS is independent of $v$. Consequently, the left- and right-rand sides of \eqref{vertex_IRF_Xl_7} are equal to a function of the integer-valued argument $l$, denote it by $-\alpha c_l$. Then $\gamma_{l+1}(v) = \alpha(v-(c_l+1)), \; \gamma_l(u) = \alpha(u-c_l).$
	Assuming $c_l = l + t$, where $t$ is an arbitrary real-valued constant, we obtain
	\begin{equation} \label{Xl}
		\psi(u)^l_{l+1}=X_l(u,t)=\begin{pmatrix}
			1\\ \alpha(u-l-t)
		\end{pmatrix}.
	\end{equation}
	The substitution of \eqref{Xl} into \eqref{vertex_IRF_Xl_3_1} gives
	\begin{equation} \label{lambda_l}
		W^{(1,1)}\left(\left.\begin{array}{cc}
			l & l+1\\
			l+1 & l+2
		\end{array}\right|u-v\right)=\lambda_l(u-v)=u-v+1.
	\end{equation}
	It is not difficult to verify that \eqref{Xl} and \eqref{lambda_l} satisfy the relations (\ref{vertex_IRF_Xl_3_2}~--~\ref{vertex_IRF_Xl_3_4}).
	
	Choosing in (\ref{vertex_IRF_fundamental}) $a = l, b = l-1, c= l-2$ one finds the intertwining vectors $\psi^{(1)}(u)^l_{l-1}$ by similar calculations:
	\begin{equation} \label{Yl}
		\psi(u)^l_{l-1}=\begin{pmatrix}
			1\\ \alpha(u+l+s)
		\end{pmatrix},
	\end{equation}
	where $s$ is a free parameter.
	
	It is possible to show that \eqref{vertex_IRF_fundamental} holds for all other possible combinations of $a$, $b$ and $c$ and find other SOS statistical weights.
	
	Thus for the $R$-matrix \eqref{R_7V_rat} we have obtained the intertwining vectors
	\begin{equation} \label{intertwining_vectors}
		\psi^{(1)}(u)^l_{l+1}=\begin{pmatrix}
			1\\ \alpha(u-l-t)
		\end{pmatrix},\quad \psi^{(1)}(u)^l_{l-1}=\begin{pmatrix}
			1\\ \alpha(u+l+s)
		\end{pmatrix}
	\end{equation}
	and the corresponding SOS Boltzmann weights
	\begin{align}
		\label{IRF_weights_pm_2} &
		W^{(1,1)}\left(\left.\begin{array}{cc}
			l\pm 2 & l\pm 1\\
			l\pm 1 & l
		\end{array}\right|u\right)=u+1, \\
		\label{IRF_weights_pm_pm} &
		W^{(1,1)}\left(\left.\begin{array}{cc}
			l & l\pm 1\\
			l\pm 1 & l
		\end{array}\right|u\right)=\frac{\mp u+l+w}{l+w} \\
		\label{IRF_weights_pm_mp} &
		W^{(1,1)}\left(\left.\begin{array}{cc}
			l & l\pm 1\\
			l\mp 1 & l
		\end{array}\right|u\right)=\frac{u\left(l\pm 1+w\right)}{l+w},
	\end{align}
	where $w=\nicefrac{1}{2}(s+t)$.
	If $w\notin \mathbb{Z}$, the vectors from \eqref{intertwining_vectors} obey the condition of proposition~\ref{YBE_SOS_from_vertex_prop}. Then $W^{(1,1)}$ from (\ref{IRF_weights_pm_2} -- \ref{IRF_weights_pm_mp}) satisfy the SOS Yang-Baxter equation \eqref{YBE_IRF}.
	
	According to the fusion procedure the vectors $\psi^{(k)}(u)^a_b$ are constructed from $\psi^{(1)}(u)^a_b$ by means of the formula (\ref{fused_vector}).

	\medskip

	\begin{prop} \label{vector_fusion_prop}
		The intertwining vectors
		\begin{equation} \label{fused_vector_proof}
			\psi^{(k)}(u)^a_b = \Pi_{1\ldots k} \, \psi^{(1)}(u+k-1)^{c_0}_{c_1}\otimes\ldots\otimes\psi^{(1)}(u+1)^{c_{k-2}}_{c_{k-1}}\otimes\psi^{(1)}(u)^{c_{k-1}}_{c_k}
		\end{equation}
		(where $c_0 = a, c_k = b$), obtained from \eqref{intertwining_vectors} with the help of the formula \eqref{fused_vector}, do not depend on the choice of $c_1, c_2,\ldots, c_{k-1}$ provided that
		\begin{equation} \label{iv_condition}
			|c_1-a|=|c_2-c_1|=\ldots=|b-c_{k-1}|=1.
		\end{equation}
	\end{prop}
	\begin{proof}
		$ $\\
		If the condition \eqref{iv_condition} holds, then $c_{i+1}=c_i+\sigma_i, \; \sigma_i=\pm 1, \; i=0,1,\ldots,k-1.$
		Denote by $k_+$ the number of values $+1$ among $\{\sigma_i\}$ (i.e. the number of vectors of the type \eqref{Xl} in the RHS of \eqref{fused_vector_proof}), and by $k_-$ -- the number of values $(-1)$ (i.e. the number of vectors of the type \eqref{Yl} in the RHS of \eqref{fused_vector_proof}). Note that
		\begin{equation} \nonumber
			\begin{cases}
				k_+-k_-=b-a \\ k_++k_-=k
			\end{cases}
			\Leftrightarrow
			\begin{cases}
				k_+=\frac{1}{2}(k+b-a) \\ k_-=\frac{1}{2}(k-b+a)
			\end{cases} .
		\end{equation}
		Let $\sigma_i=+1$ for $i=i_1,i_2,\ldots,i_{k_+}$. In this case $c_{i_p}=a-i_p+2(p-1), \; p=1,2,\ldots,k_+.$
		Then substituting $c_{i_p}$ into \eqref{intertwining_vectors} we obtain that the vector
		\begin{equation} \label{psi_p}
			\psi(u+k-i_p-1)^{c_{i_p}}_{c_{i_p}+1}=\begin{pmatrix}
				1 \\ \alpha(u+k-a-2p+1-t)
			\end{pmatrix}
		\end{equation}
		is the $(i_p+1)$-th factor in the tensor product from \eqref{fused_vector_proof} (where $p=1,2,\ldots,k_+$).
		Let $\sigma_j=-1$ for $j=j_1,j_2,\ldots,j_{k_-}$. In this instance $c_{j_q}=a+j_q-2(q-1), \; q=1,2,\ldots,k_-.$ Then substituting $c_{j_q}$ into \eqref{intertwining_vectors} one obtains that the vector
		\begin{equation} \label{psi_m}
			\psi(u+k-j_q-1)^{c_{j_q}}_{c_{j_q}-1}=\begin{pmatrix}
				1 \\ \alpha(u+k+a-2q+1+s)
			\end{pmatrix}.
		\end{equation}
		is the $(j_q+1)$-th factor in the tensor product \eqref{fused_vector_proof} (here $q=1,2,\ldots,k_-$).
		It can be seen that every vector in the tensor product from the RHS of \eqref{fused_vector_proof} depends only on its number among the vectors of type \eqref{Xl} or \eqref{Yl}. Consequently, if \eqref{iv_condition} holds, then the RHS of \eqref{fused_vector_proof} is the tensor product of $k_+$ vectors \eqref{psi_p} and $k_-$ vectors \eqref{psi_m} independently of how one chooses $c_1,c_2\ldots,c_{k-1}$. Only the order of their tensor multiplication depends on the choice of $c_1,c_2\ldots,c_{k-1}$. But the result of how the symmetrizer $\Pi_{1\ldots k}$ acts on this tensor product is independent of the factors' order.
	\end{proof}

	\medskip

	From the constructed SOS model (\ref{IRF_weights_pm_2} -- \ref{IRF_weights_pm_mp}) it is possible to obtain the model considered in the paper \cite{zapiski}. Its state parameters of vertices take nonnegative integer values, the Boltzmann weights obey the restriction \eqref{IRF_model_condition}, and non-zero statistical weights have the form
	\begin{align}
		\label{W01} &
		W_0^{(1,1)}\left(\left.\begin{array}{cc}
			l\pm 2 & l\pm 1\\
			l\pm 1 & l
		\end{array}\right|u\right)
		= u + 1, \quad
		W_0^{(1,1)}\left(\left.\begin{array}{cc}
			l & l\pm 1\\
			l\pm 1 & l
		\end{array}\right|u\right)
		= \frac{\mp u+l+1}{l+1}, \\
		\label{W02} & 
		W_0^{(1,1)}\left(\left.\begin{array}{cc}
			l & l\pm 1\\
			l\mp 1 & l
		\end{array}\right|u\right)
		= \frac{u\sqrt{l \left(l+2\right)}}{l+1}.
	\end{align}
	In order to do this modify the intertwining vectors and introduce
	\begin{equation} \label{intertwining_vectors_tilde}
		\tilde{\psi}^{(1)}(u)^a_b=\frac{1}{F(a,b)}\psi^{(1)}(u)^a_b,
	\end{equation}
	where $F(a,b)$ is a function of two integer-valued arguments set by the conditions
	\begin{equation} \nonumber
		|a-b|\neq 1\Rightarrow F(a,b)=0,\quad \frac{F(l,l+1)}{F(l,l-1)}=\left(\frac{l+1+w}{l-1+w}\right)^{1/4},\quad F(a,b)=F(b,a).
	\end{equation}
	Substituting $\psi^{(1)}(u)^a_b$ from \eqref{intertwining_vectors_tilde} into \eqref{vertex_IRF_fundamental} one obtains
	\begin{equation} \nonumber
		\check{R}^{(1,1)}(u-v) \, \tilde{\psi}^{(1)}(u)^a_b\otimes\tilde{\psi}^{(1)}(v)^b_c =
		\sum\limits_{b^\prime} \tilde{\psi}^{(1)}(v)^a_{b^\prime}\otimes\tilde{\psi}^{(1)}(u)^{b^\prime}_c \,
		\widetilde{W}^{(1,1)}\left(\left.\begin{array}{cc}
			a & b\\
			b^\prime & c
		\end{array}\right|u-v\right),
	\end{equation}
	where
	\begin{equation} \label{W_transform}
		\widetilde{W}^{(1,1)}\left(\left.\begin{array}{cc}
			a & b\\
			b^\prime & c
		\end{array}\right|u-v\right) =
		\frac{F(a,b^\prime)F(b^\prime,c)}{F(a,b)F(b,c)} \,
		W^{(1,1)}\left(\left.\begin{array}{cc}
			a & b\\
			b^\prime & c
		\end{array}\right|u-v\right) ,
	\end{equation}
	namely
	\begin{align}
		\label{IRF_weights_tilde} & \widetilde{W}^{(1,1)}\left(\left.\begin{array}{cc}
			l\pm 2 & l\pm 1\\
			l\pm 1 & l
		\end{array}\right|u\right)
		= u + 1, \quad
		\widetilde{W}^{(1,1)}\left(\left.\begin{array}{cc}
			l & l\pm 1\\
			l\pm 1 & l
		\end{array}\right|u\right)
		= \frac{\mp u+l+w}{l+w}, \\
		\label{IRF_weights_tilde_1} & \widetilde{W}^{(1,1)}\left(\left.\begin{array}{cc}
			l & l\pm 1\\
			l\mp 1 & l
		\end{array}\right|u\right)
		= \frac{u\sqrt{\left(l-1+w\right)\left(l+1+w\right)}}{l+w}
	\end{align}
	(the transform of SOS Boltzmann weights similar to \eqref{W_transform} was considered in \cite{vega}).
	For $w\notin \mathbb{Z}$ the vectors $\tilde{\psi}^{(1)}(u)^a_b$ obey the condition of the proposition~\ref{YBE_SOS_from_vertex_prop}. Therefore, the Boltzmann weights (\ref{IRF_weights_tilde} -- \ref{IRF_weights_tilde_1}) satisfy the SOS Yang-Baxter equation \eqref{YBE_IRF}. The model (\ref{W01} -- \ref{W02}) is obtained if one considers (\ref{IRF_weights_tilde} -- \ref{IRF_weights_tilde_1}) with nonnegative state parameters and $w=1$.
	
	Let us show that from the relation \eqref{YBE_IRF} for the model (\ref{IRF_weights_tilde} -- \ref{IRF_weights_tilde_1}) with integer-valued state parameters follows \eqref{YBE_IRF} for the model (\ref{W01} -- \ref{W02}). For this purpose consider \eqref{YBE_IRF} for $\widetilde{W}^{(1,1)}$ with $a,b,c,d,e,f\geq 0$ and $w$ close to $1$, but not integer. The terms with negative $g$, which do not correspond to the model (\ref{W01} -- \ref{W02}), can appear in \eqref{YBE_IRF} only in two cases: if $a=c=e=1,\: b=d=f=0$ or if $a=c=e=0,\: b=d=f=1$. In these cases take the limit $w\to 1$ in both sides of \eqref{YBE_IRF}. The limits of both sides exist, because every Boltzmann weight $\widetilde{W}$ has a finite limit (denominators are not equal to zero even for the mentioned summands with negative $g$), and the sums are finite in the bots sides of \eqref{YBE_IRF}. From \eqref{IRF_weights_tilde_1} it follows that the limit of
	$\widetilde{W}^{(1,1)}\left(\left.\begin{array}{cc}
		0 & \pm 1\\
		\mp 1 & 0
	\end{array}\right|u\right)$
	as $w\to 1$ is zero, so the limit of mentioned unwanted terms with negative $g$ in \eqref{YBE_IRF} is also equal to zero, and we are left with only the summands corresponding to the model (\ref{W01} -- \ref{W02}).

	\section{Realization on the space of polynomials}
	\label{polynomial_sect}
	
	Our final objective is to find the Boltzmann weights $W^{(n,m)}$ of the SOS model corresponding to the seven-vertex lattice model. To do this we need to construct the operator $R^{(n,m)}$ by means of the fusion procedure. It is convenient to realize the spaces on which it acts as the spaces of polynomials of order less than a fixed number. With this approach $R^{(n,m)}$ can be written in a compact form in terms of a difference operator.
	
	$R$-operator has the form $\Pi_{1\ldots n}T$, where $T$ is some operator on $(\mathbb{C}^2)^{\otimes n}$. Intertwining vectors have the form $\Pi_{1\ldots n}\Psi$, where $\Psi$ is a tensor in $(\mathbb{C}^2)^{\otimes n}$. We realize the intertwining vectors and the $R$-operator on the space of polynomials using the fact that $S^n\mathbb{C}^2$ is isomorphic to the space of homogeneous polynomials of degree $n$ in two variables $\lambda_1,\lambda_2$. Isomorphism $S$ reads
	\begin{equation} \nonumber
		S: \; \Pi_{1\ldots n} \, e^{i_1}\otimes\ldots\otimes e^{i_n} \mapsto \lambda_{i_1}\ldots\lambda_{i_n}, \qquad i_1,\ldots,i_n\in\{1,2\}.
	\end{equation}
	Tensor $\Psi$ can be written in the basis of $(\mathbb{C}^2)^{\otimes n}$ (summation over repeated indices is implied): $\Psi=\Psi_{i_1\ldots i_n} \, e^{i_1}\otimes\ldots\otimes e^{i_n}.$
	Then the polynomial realization of $\Pi_{1\ldots n}\Psi$ has the form
	\begin{equation} \nonumber
		[S\Pi_{1\ldots n}\Psi](\lambda) = \Psi_{i_1\ldots i_n}S\left(\Pi_{1\ldots n} \, e^{i_1}\otimes\ldots\otimes e^{i_n}\right) = \Psi_{i_1\ldots i_n}\lambda_{i_1}\ldots\lambda_{i_n}.
	\end{equation}
	Hence using the formulae \eqref{fused_vector}, \eqref{psi_p}, \eqref{psi_m} one obtains that the intertwining vectors in the polynomial representation have the form:
	\begin{align}
		\nonumber \psi^{(n)}(\lambda_1,\lambda_2|u)^a_b = & \,
		\prod_{p=1}^{n_+}[\lambda_1+\alpha\lambda_2(u+n-a-2p+1-t)] \\
		\label{intertwining_vector_polynomial} & \times\prod_{q=1}^{n_-}[\lambda_1+\alpha\lambda_2(u+n+a-2q+1+s)],
	\end{align}
	where $n_+=\nicefrac{1}{2}(n+b-a), \, n_-=\nicefrac{1}{2}(n-b+a).$
	
	The action of an operator $T$ on a tensor $\Psi$ reads: $(T\Psi)_{i_1\ldots i_n}=T_{i_1\ldots i_n}^{j_1\ldots j_n}\Psi_{j_1\ldots j_n}.$
	In the realization on the space of polynomials the action of operators has the following form \cite[section~2.5]{der_ch}:
	\begin{equation} \nonumber
		[S\Pi_{1\ldots n}T\Pi_{1\ldots n}\Psi](\lambda) = 
		\frac{1}{n!}\left.T(\lambda,\partial_\mu)[S\Pi_{1\ldots n}\Psi](\mu)\right|_{\mu=0},
	\end{equation}
	where $T(\lambda,\mu)\equiv\lambda_{j_1}\ldots\lambda_{j_n}T_{j_1\ldots j_n}^{i_1\ldots i_n}\mu_{i_1}\ldots\mu_{i_n}$
	is called the symbol of the operator $T$.
	
	Let us find the symbol of the operator $R^{(n,m)}(u)$ \eqref{R_k_n} acting on $V^{(n,m)}$. At first one should find the symbol $R^{(n,1)}_{1\ldots n,\bar{j}}(u|\lambda,\mu)$ of the operator $R^{(n,1)}_{1\ldots n,\bar{j}}(u)$ \eqref{R_k_1}. This symbol is an operator in the space $V_{\bar{j}}$.
	\begin{equation} \nonumber
		R^{(n,1)}_{1\ldots n,\bar{j}}(u|\lambda,\mu)=R^{(1,1)}_{1\bar{j}}(u+n-1|\lambda,\mu)\ldots R^{(1,1)}_{n-1\bar{j}}(u+1|\lambda,\mu)R^{(1,1)}_{n\bar{j}}(u|\lambda,\mu),
	\end{equation}
	where $R^{(1,1)}_{i\bar{j}}(u+n-i|\lambda,\mu)$ is the symbol of the operator $R^{(1,1)}_{i\bar{j}}(u+n-i)$ in the space $V_i$, it is an operator in the space $V_{\bar{j}}$. Proceeding from variables $\lambda$ to the representation on polynomials in one variable $z$: $\lambda_1=-z$, $\lambda_2=1$ one obtains after some algebra \cite{VA_19}
	\begin{align}
		\nonumber &
		R^{(n,1)}_{1\ldots n,\bar{j}}(u|\lambda,\mu) \\
		\nonumber &
		=
		\begin{pmatrix}
			u\Delta_++z\alpha^{-1}\Delta_- & -\alpha^{-1}\Delta_-\\
			z^2\alpha^{-1}\Delta_--nz\Delta_+-\alpha u(u+n)\Delta_- & (u+n)\Delta_+-z\alpha^{-1}\Delta_-
		\end{pmatrix}(\mu_2-\mu_1z)^n,
	\end{align}
	where the operators $\Delta_+$ and $\Delta_-$ are defined as
	\begin{equation} \nonumber
		[\Delta_\pm f](z) \equiv \frac{1}{2}(f(z+\alpha) \pm f(z-\alpha)) \, .
	\end{equation}
	Passing from the symbol of the operator to the operator itself one finds
	\begin{equation} \label{R_k_1_mat_form}
		R^{(n,1)}_{1\ldots n,\bar{j}}(u)=\begin{pmatrix}
			u\Delta_++z\alpha^{-1}\Delta_- & -\alpha^{-1}\Delta_-\\
			z^2\alpha^{-1}\Delta_--nz\Delta_+-\alpha u(u+n)\Delta_- & (u+n)\Delta_+-z\alpha^{-1}\Delta_-
		\end{pmatrix}.
	\end{equation}
	The matrix $R^{(n,1)}_{1\ldots k,\bar{j}}(u)$ in \eqref{R_k_1_mat_form} is written in the basis of the space $V_{\bar{j}} \simeq \mathbb{C}^2$, the operators $\Delta_+$ and $\Delta_-$ act on the space of polynomials of degree less or equal than $n$ which is isomorphic to $\mathrm{Sym}(V_1\otimes\ldots\otimes V_n) \simeq S^n\mathbb{C}^2 .$
	
	The next step is the determination of $R^{(n,m)}(u)$ and of its symbol $R^{(n,m)}(u|\lambda,\mu)$, where the space of polynomials in $\lambda_1, \lambda_2$ is associated to $\mathrm{Sym}V_{\bar{1}}\otimes V_{\bar{2}}\otimes\ldots\otimes V_{\bar{m}}.$ In order to do this one needs to construct the symbol $\Lambda(u|\lambda,\mu)$ of the operator $R^{(n,1)}_{1\ldots n,\bar{j}}(u)$, where $V_{\bar{j}}$ interpreted as the space of polynomials \cite{der_ch}.
	\begin{align} 
		\nonumber \Lambda(u|\lambda,\mu) = \lambda_{i_1}\left(R^{(n,1)}_{1\ldots n,\bar{j}}(u)\right)_{i_1}^{j_1}\mu_{j_1} = & \,
		u(\lambda_1+\lambda_2z)\mu_1\Delta_+-(\lambda_1+\lambda_2z)(\mu_2-\mu_1z)\alpha^{-1}\Delta_- \\
		\label{Lambda} & +(u+n)\lambda_2(\mu_2-\mu_1z)\Delta_+-\alpha u(u+n)\mu_1\lambda_2\Delta_-,
	\end{align}
	where $\left(R^{(n,1)}_{1\ldots n,\bar{j}}(u)\right)_{i_1}^{j_1}$ is the element $(i_1,j_1)$ of the matrix \eqref{R_k_1_mat_form}. The symbol $R^{(n,m)}(u|\lambda,\mu)$ is expressed in terms of $\Lambda(u|\lambda,\mu)$ in the following way \cite[section~2.6]{der_ch}:
	\begin{equation} \label{R_k_n_symb_1}
		R^{(n,m)}(u|\lambda,\mu) =
		\Lambda(u|\lambda,\mu) \, \Lambda(u-1|\lambda,\mu) \ldots \Lambda(u-m+1|\lambda,\mu).
	\end{equation}
	
	The vertex-SOS correspondence \eqref{vertex_IRF} in the polynomial representation with $u = 0,\:v = -u$ takes the form
	\begin{align} 
		\nonumber &
		\left.\left[R^{(n,m)}(u|\lambda,\partial_\mu)\psi^{(m)}(\mu_1,\mu_2|-u)^b_c\right]\right|_{\mu=0}\psi^{(n)}(z|0)^a_b \\
		\label{vertex_IRF_polynomials} &
		= \sum\limits_{b^\prime}\psi^{(n)}(z|0)^{b^\prime}_c\psi^{(m)}(\lambda_1,\lambda_2|-u)^a_{b^\prime}
		W^{(n,m)}\left(\left.\begin{array}{cc}
			a & b\\
			b^\prime & c
		\end{array}\right|u\right),
	\end{align}
	where the intertwining vectors $\psi^{(n)}(z|u)^a_b$ are obtained from \eqref{intertwining_vector_polynomial} by the substitution $\lambda_1=-z,\lambda_2=1:$
	\begin{equation} \label{intertwining_vector_polynomial_z}
		\psi^{(n)}(z|u)^a_b = (-1)^n\prod_{p=1}^{n_+}[z-\alpha(u+n-a-2p+1-t)]
		\prod_{q=1}^{n_-}[z-\alpha(u+n+a-2q+1+s)],
	\end{equation}
	and $n_\pm = \nicefrac{1}{2}\left(n\pm(b-a)\right)$.
	Considering only the coefficients of $\lambda_1^m$ in the both sides of \eqref{vertex_IRF_polynomials} one obtains from \eqref{Lambda} and \eqref{R_k_n_symb_1} that
	\begin{equation} \label{defining_relation}
		O_m(u,b,c) \, \psi^{(n)}(z|0)^a_b=\sum\limits_{b^\prime}W^{(n,m)}\left(\left.\begin{array}{cc}
			a & b\\
			b^\prime & c
		\end{array}\right|u\right)\psi^{(n)}(z|0)^{b^\prime}_c,
	\end{equation}
	where the operator $O_m(u,b,c)$ has the form
	\begin{equation} \label{O_m}
		O_m(u,b,c) = \left.\Lambda^\prime(u|\partial_\mu)\Lambda^\prime(u-1|\partial_\mu)\ldots\Lambda^\prime(u-m+1|\partial_\mu)\psi^{(m)}(\mu_1,\mu_2|-u)^b_c\right|_{\mu=0},
	\end{equation}
	\begin{equation} \label{Lambda_prime}
		\Lambda^\prime(u|\mu)\equiv-\alpha^{-1}[(\mu_2-\mu_1z)\Delta_--\alpha u\mu_1\Delta_+].
	\end{equation}
	
	Let us show that
	\begin{align} 
	\nonumber O_m(u,b,c) = & \, \alpha^{-m}\gamma\left(z-u_1,-u+m_+\right)\Delta_-^{m_+}\gamma\left(z-u_1, u\right) \\
	\label{O_m_1} & \times \gamma\left(z-u_2, -u+m\right) \Delta_-^{m_-}
	\gamma\left(z-u_2, u-m_+\right),
	\end{align}
	where $m_\pm = \nicefrac{1}{2}\left(m\pm(c-b)\right),$
	\begin{equation} \label{u1u2}
		u_1 = \alpha\left[-u+\frac{1}{2}(m-b-c)-t\right], \qquad
		u_2 = \alpha\left[-u+\frac{1}{2}(m+b+c)+s\right],
	\end{equation}
	\begin{equation} \label{gamma_function}
	 \gamma(z, p) = (2\alpha)^p \, \frac{\Gamma\left(\frac{z}{2\alpha}+\frac{1}{2}+\frac{p}{2}\right)}{\Gamma\left(\frac{z}{2\alpha}+\frac{1}{2}-\frac{p}{2}\right)}, \quad p \in \mathbb{C},
	\end{equation}
	and $\Gamma$ is the gamma function.
	In the derivation of \eqref{O_m_1} the ``star-triangle'' relation
	\begin{equation} \label{STR}
		\gamma(z,k) \, \Delta_-^{k+l} \, \gamma(z,l) = \Delta_-^l \, \gamma(z,k+l) \, \Delta_-^k \, .
	\end{equation}
	plays an important role.
	It holds for nonnegative integer numbers $k$ and $l$ and can be proven by induction with the help of the formulae
	\begin{align}
	\label{Delta_gamma} & \Delta_- \, \gamma(z,p) = \gamma(z,p-1)[z\Delta_- + p\alpha\Delta_+], \\
	\label{gamma_Delta} & \gamma(z,p) \, \Delta_- = [\Delta_-z - p\alpha\Delta_+] \, \gamma(z,p-1)
	\end{align}
	($\Delta_-$ acts on the variable $z$), which follow from the definition of the operators $\Delta_+$ and $\Delta_-$.
	
	By means of \eqref{Delta_gamma} and \eqref{gamma_Delta} it is possible to show that the RHS of \eqref{O_m_1} is equal to
	\begin{align}
	\nonumber & \alpha^{-m}\prod\limits_{l=0}^{\overset{m_+-1}{\longleftarrow}}\left\{\left[z-\alpha\left(-u+\frac{m-b-c}{2}-t\right)\right]\Delta_- + \alpha(u-l)\Delta_+\right\} \\
	\label{O_m_2} & \times\prod\limits_{l'=0}^{\overset{m_--1}{\longleftarrow}}\left\{\left[z-\alpha\left(-u+\frac{m+b+c}{2}+s\right)\right]\Delta_- + \alpha(u-m_+-l')\Delta_+\right\},
	\end{align}
	where for $N_2 \geq N_1$
	\begin{equation} \nonumber
		\prod\limits_{l=N_1}^{\overset{N_2}{\longleftarrow}} A_l \equiv A_{N_2} A_{N_2-1} \ldots A_{N_1+1} A_{N_1} \, .
	\end{equation}
	Then the RHS of \eqref{O_m_1} is a polynomial in $u$ of degree $m$. Moreover, from \eqref{O_m} it follows that $O_m(u,b,c)$ is also a polynomial in $u$ of degree $m$. Thus, if one proves \eqref{O_m_1} for some set of $m+1$ values of $u$, for example $\{0,1,2,\ldots,m\}$, then it will hold for any $u$.

	\medskip

	\begin{prop} \label{O_m_prop}
		The equality \eqref{O_m_1} holds for $u\in N_m=\{0,1,2,\ldots,m\}$.
	\end{prop}
	\begin{proof}
		$ $\\
		Consider $u\in N_m$. With the help of relations
		\begin{align}
			\nonumber & (\mu_2-\mu_1z)\Delta_-^p=\Delta_-^{p-1}[(\mu_2-\mu_1z)\Delta_-+\mu_1(p-1)\alpha\Delta_+], \\
			\nonumber & \Delta_-^p(\mu_2-\mu_1z)=[(\mu_2-\mu_1z)\Delta_--\mu_1p\alpha\Delta_+]\Delta_-^{p-1},
		\end{align}
		following from the identity $\Delta_- z = z \Delta_- + \alpha\Delta_+$, one can deduce from \eqref{Lambda_prime} that
		\begin{equation} \label{Lambda_product}
			\Lambda^\prime(u|\mu)\Lambda^\prime(u-1|\mu)\ldots\Lambda^\prime(u-m+1|\mu)=(-\alpha)^{-m}\Delta_-^u(\mu_2-\mu_1z)^m\Delta_-^{m-u}.
		\end{equation}
		The substitution of \eqref{Lambda_product} into \eqref{O_m} gives
		\begin{equation} \label{O_m_3}
			O_m(u,b,c)=(-\alpha)^{-m}\Delta_-^u\psi^{(m)}(z|-u)^b_c\Delta_-^{m-u}.
		\end{equation}
		Expressing $\psi^{(m)}(z|-u)^b_c$ in terms of the function $\gamma$ \eqref{gamma_function} by means of \eqref{intertwining_vector_polynomial_z}:
		\begin{equation} \nonumber
			\psi^{(m)}(z|-u)^b_c = (-1)^m\gamma\left(z-u_1,m_+\right) \gamma\left(z-u_2,m_-\right) ,
		\end{equation}
		where $u_1$ and $u_2$ are defined in \eqref{u1u2}, and substituting the result into \eqref{O_m_3} one obtains that
		\begin{equation} \label{O_m_4}
			O_m(u,b,c) = \alpha^{-m} \Delta_-^u \, \gamma\left(z-u_1, m_+\right) \gamma\left(z-u_2,m_-\right)\Delta_-^{m-u}.
		\end{equation}
		Comparing \eqref{O_m_4} and \eqref{O_m_1}, we conclude that it remains to prove the following:
		\begin{align}
			\nonumber & \Delta_-^u \, \gamma\left(z-u_1, m_+\right) \gamma\left(z-u_2,m_-\right)\Delta_-^{m-u} \\
			\label{O_m_5} &
			= \gamma\left(z-u_1,-u+m_+\right)\Delta_-^{m_+}\gamma\left(z-u_1, u\right) \gamma\left(z-u_2, -u+m\right) \Delta_-^{m_-}
			\gamma\left(z-u_2, u-m_+\right) .
		\end{align}
		Consider the case $u\leq m_+$. From the obvious property $\gamma(z, p) \gamma(z, -p) = 1$ it follows that \eqref{O_m_5} is equivalent to
		\begin{align}
			\nonumber & \Delta_-^u \, \gamma\left(z-u_1, m_+\right) \gamma\left(z-u_2,m_-\right)\Delta_-^{m-u} \gamma\left(z-u_2, m_+-u\right) \\
			\label{O_m_6} & = \gamma\left(z-u_1,-u+m_+\right)\Delta_-^{m_+}\gamma\left(z-u_1, u\right) \gamma\left(z-u_2, -u+m\right) \Delta_-^{m_-} \, .
		\end{align}
		The proof of \eqref{O_m_6} involves \eqref{STR}:
		\begin{align}
			\nonumber & \Delta_-^u \, \gamma\left(z-u_1, m_+\right) \gamma\left(z-u_2,m_-\right)\Delta_-^{m-u} \gamma\left(z-u_2, m_+-u\right) \\
			\nonumber & = \Delta_-^u \, \gamma\left(z-u_1, m_+\right) \Delta_-^{m_+-u} \gamma\left(z-u_2, m-u\right) \Delta_-^{m_-} \\
			\nonumber & = \gamma\left(z-u_1, m_+-u\right) \Delta_-^{m_+} \gamma\left(z-u_1, u\right) \gamma\left(z-u_2, m-u\right) \Delta_-^{m_-} \, .
		\end{align}
		In the case $u\geq m_+$ the relation \eqref{O_m_5} is equivalent to
		\begin{align}
			\nonumber & \gamma\left(z-u_1,u-m_+\right) \Delta_-^u \, \gamma\left(z-u_1, m_+\right) \gamma\left(z-u_2,m_-\right)\Delta_-^{m-u} \\
			\label{O_m_7} & = \Delta_-^{m_+}\gamma\left(z-u_1, u\right) \gamma\left(z-u_2, -u+m\right) \Delta_-^{m_-}
			\gamma\left(z-u_2, u-m_+\right) .
		\end{align}
		The equality \eqref{O_m_7} can be proven with the help of \eqref{STR} similarly to \eqref{O_m_6}.
	\end{proof}

	\medskip

	Having proved the proposition~\ref{O_m_prop}, we have proven \eqref{O_m_1} for all $u\in\mathbb{C}$. Hence, from \eqref{defining_relation}, \eqref{O_m_1} and \eqref{O_m_2} follows the relation by means of which it is possible to find the Boltzmann weights $W^{(n,m)}$:
	\begin{align}
	\nonumber & (-\alpha)^{-m}\prod\limits_{l=0}^{\overset{m_+-1}{\longleftarrow}}\left\{\left[\alpha\left(-u+\frac{m-b-c}{2}-t\right)-z\right]\Delta_--\alpha(u-l)\Delta_+\right\} \\
	\nonumber & 
	\times\prod\limits_{l'=0}^{\overset{m_--1}{\longleftarrow}}\left\{\left[\alpha\left(-u+\frac{m+b+c}{2}+s\right)-z\right]\Delta_--\alpha(u-m_+-l')\Delta_+\right\}\psi^{(n)}(z|0)^a_b \\
	\label{defining_relation_1} & = \sum\limits_{b^\prime}W^{(n,m)}\left(\left.\begin{array}{cc}
	a & b\\
	b^\prime & c
	\end{array}\right|u\right)\psi^{(n)}(z|0)^{b^\prime}_c \, ,
	\end{align}
	where $m_\pm = \nicefrac{1}{2}(m\pm(c-b))$. The weights $W^{(n,m)}$ are uniquely determined from \eqref{defining_relation_1} and satisfy the SOS Yang-Baxter equation \eqref{YBE_IRF} when the following proposition holds.

	\medskip

	\begin{prop} \label{vectors_linear_independence_prop}
		For any $n \in \mathbb{N}$, any $u \in \mathbb{C}$ and any $a, c \in \mathbb{Z}$ the sets of intertwining vectors
		\begin{equation} \label{iv}
			\{\psi^{(n)}(u)^a_b\}_{b\in\{a-n+2l, \; l=0,1,\ldots,n\}}, \qquad 
			\{\psi^{(n)}(u)^b_c\}_{b\in\{c-n+2l, \; l=0,1,\ldots,n\}} \, ,
		\end{equation}
		constructed from \eqref{intertwining_vectors} with the help of fusion procedure \eqref{fused_vector}, are linearly independent if $w = \nicefrac{1}{2}(s+t) \notin \mathbb{Z}$.
	\end{prop}
	\begin{proof}
		$ $\\
		The proof uses the representation in the space of polynomials. Denote
		\begin{equation} \nonumber
			p_l(z) \equiv \psi^{(n)}(z|u)^a_{a-n+2l} = (-1)^n\prod\limits_{i=1}^{l}(z-a_i)\prod\limits_{j=l+1}^{n}(z-b_j),\quad l=0,1,2,\ldots,n,
		\end{equation}
		where $a_i = \alpha(u+n-a-2i+1-t), \; b_j = \alpha(u+n+a-2(n+1-j)+1+s), \; 1\leq i,j\leq n$
		(the expression for $\psi^{(n)}(z|u)^a_{a-n+2l}$ is obtained from \eqref{intertwining_vector_polynomial_z}). Note that, since by the data $(s+t)\notin 2\mathbb{Z}$, the following holds:
		\begin{equation} \label{a_i_neq_b_j}
			\forall i,j \quad a_i\neq b_j \, .
		\end{equation}
		Furthermore,
		\begin{equation} \label{p_l_a_i_b_j}
			p_l(a_i)=0, \; i\leq l, \qquad p_l(b_j)=0, \; j>l.
		\end{equation}
		Consider
		\begin{equation} \label{fz0}
			0=f(z)=\sum\limits_{l=0}^n\alpha_lp_l(z),
		\end{equation}
		let us prove that $\alpha_0=\alpha_1=\ldots=\alpha_n=0$. At first, $0=f(a_1)=\alpha_0p_0(a_1),$
		consequently, $\alpha_0=0$. Then, by induction, using \eqref{p_l_a_i_b_j}, it is possible to prove that $\alpha_1=\ldots=\alpha_{n-1}=0$. So, we obtain
		\begin{equation} \label{alp0}
			\alpha_np_n(z)=0,
		\end{equation}
		therefore $\alpha_n=0$.
		
		Now denote by $p_l(z)$ the second set of vectors in \eqref{iv}, express them with the help of \eqref{intertwining_vector_polynomial_z}:
		\begin{equation} \nonumber
			p_l(z) \equiv \psi^{(n)}(z|u)^{c-n+2l}_{c} 
			= (-1)^n \prod\limits_{i=1}^{l}(z-a_i)
			\prod\limits_{j=l+1}^{n}(z-b_j),\quad l=0,1,2,\ldots,n,
		\end{equation}
		where $a_i = \alpha(u+c+2i-1+s), \; b_j = \alpha(u-c+2(n+1-j)-1-t), \; 1\leq i,j\leq n.$
		The numbers $a_i$ and $b_j$ satisfy \eqref{a_i_neq_b_j} and \eqref{p_l_a_i_b_j}, which means the further proof is similar to (\ref{fz0} -- \ref{alp0}).
	\end{proof}

	\section{Computation of Boltzmann weights of SOS models} \label{Wnm_section}

	\subsection{Calculation of Boltzmann weights $W^{(n,1)}$}
	
	Consider the relation \eqref{defining_relation_1} with $m=1$, $b=c+1$:
	\begin{align}
	\nonumber & -\alpha^{-1}\{[\alpha(-u+c+1+s)-z]\Delta_--\alpha u\Delta_+\}\psi^{(n)}(z|0)^{c+k+1}_{c+1} \\
	\label{defining_relation_W_n_1_plus} & =\sum\limits_{\sigma=\pm}
	W^{(n,1)}\left(\left.\begin{array}{cc}
		c+k+1 & c+1\\
		c+k+1+\sigma & c
	\end{array}\right|u\right)
	\psi^{(n)}(z|0)^{c+k+1+\sigma}_c,
	\end{align}
	where $k\in\{-n+2j,\: j=0,1,\ldots,n\}.$ Using \eqref{intertwining_vector_polynomial_z} one obtains by direct calculation that
	\begin{align}
	\nonumber &
	W^{(n,1)}\left(\left.\begin{array}{cc}
	c+k+1 & c+1\\
	c+k+2 & c
	\end{array}\right|u\right)
	= \frac{n_-(k)\left(c+1-n_-(k)+w-u\right)}{c+k+1+w}, \\
	\label{W_n_1_c_plus_1} &
	W^{(n,1)}\left(\left.\begin{array}{cc}
	c+k+1 & c+1\\
	c+k & c
	\end{array}\right|u\right)=\frac{\left(u+n_+(k)\right)\left(c+1+n_+(k)+w\right)}{c+k+1+w} ,
	\end{align}
	where $n_\pm(k) = \nicefrac{1}{2}(n\pm k),$ $w=\nicefrac{1}{2}(s+t).$
	
	Note that the upper statistical weight in \eqref{W_n_1_c_plus_1} is equal to zero when $k=n$, as it should be, in accordance with the condition \eqref{IRF_model_condition}, since the difference of lower left and lower right parameters becomes equal to $n+2>n$.
	
	Consider the relation \eqref{defining_relation_1} with $m=1$, $b=c-1$:
	\begin{align}
	\nonumber & -\alpha^{-1}\{[\alpha(-u-c+1-t)-z]\Delta_--\alpha u\Delta_+\}\psi^{(n)}(z|0)^{c+k-1}_{c-1} \\
	\label{defining_relation_W_n_1_minus} &
	= \sum\limits_{\sigma=\pm}
	W^{(n,1)}\left(\left.\begin{array}{cc}
	c+k-1 & c-1\\
	c+k-1+\sigma & c
	\end{array}\right|u\right)
	\psi^{(n)}(z|0)^{c+k-1+\sigma}_c \, .
	\end{align}
	In this case
	\begin{align} 
	\nonumber &
	W^{(n,1)}\left(\left.\begin{array}{cc}
	c+k-1 & c-1\\
	c+k-2 & c
	\end{array}\right|u\right)
	= \frac{n_+(k)\left(c-1+n_+(k)+w+u\right)}{c+k-1+w}, \\
	\label{W_n_1_c_minus_1} &
	W^{(n,1)}\left(\left.\begin{array}{cc}
	c+k-1 & c-1\\
	c+k & c
	\end{array}\right|u\right)
	= \frac{\left(u+n_-(k)\right)\left(c-1-n_-(k)+w\right)}{c+k-1+w},
	\end{align}
	where $n_\pm(k) = \nicefrac{1}{2}(n\pm k),$ $w = \nicefrac{1}{2}(s+t),$ $k \in \{-n+2j,\: j=0,1,\ldots,n\}.$
	
	Similarly to \eqref{W_n_1_c_plus_1}, the upper statistical weight in \eqref{W_n_1_c_minus_1} equals to zero when $k=-n$, in compliance with the condition \eqref{IRF_model_condition}, because the difference of lower left and lower right parameters becomes equal to $-n-2 < -n$.

	\subsection{Calculation of Boltzmann weights $W^{(n,m)}$}

	Denote in \eqref{defining_relation_1} $b=c+\mu, \; a=c+\mu+\nu, \; b^\prime=c+\mu+\nu+\mu^\prime \, ,$
	where $\mu, \mu^\prime \in \{-m+2j,j=0,1,\ldots,m\}$, $\nu\in\{-n+2j,j=0,1,\ldots,n\}$. In this notation \eqref{defining_relation_1} can be rewritten in the following form:
	\begin{align}
	\nonumber &
	(-\alpha)^{-m} \, P_1 \, P_2 \, \psi^{(n)}(z|0)^{c+\mu+\nu}_{c+\mu} \\
	\label{defining_relation_W_n_m} &
	= \sum\limits_{\mu^\prime}
	W^{(n,m)}\left(\left.\begin{array}{cc}
		c+\mu+\nu & c+\mu\\
		c+\mu+\nu+\mu^\prime & c
	\end{array}\right|u\right)
	\psi^{(n)}(z|0)^{c+\mu+\nu+\mu^\prime}_c \, ,
	\end{align}
	where $m_\pm = \frac{1}{2}(m\mp\mu),$
	\begin{equation} \nonumber
		P_1 \equiv
		\prod\limits_{l=0}^{\overset{m_+-1}{\longleftarrow}}\left\{\left[\alpha\left(-[u-l]-[c-m_++l]-t\right)-z\right]\Delta_--\alpha[u-l]\Delta_+\right\} ,
	\end{equation}
	\begin{equation} \nonumber
		P_2 \equiv
		\prod\limits_{l'=0}^{\overset{m_--1}{\longleftarrow}}\left\{\left[\alpha\left(-[u-m_+-l^\prime]+[c+\mu-l^\prime]+s\right)-z\right]\Delta_--\alpha[u-m_+-l^\prime]\Delta_+\right\} .
	\end{equation}
	
	Applying \eqref{defining_relation_W_n_1_plus} $m_-$ times one obtains that
	\begin{align}
	\nonumber &
	(-\alpha)^{-m_-} P_2 \, \psi^{(n)}(z|0)^{c+\mu+\nu}_{c+\mu} \\
	\label{second_product_action} &
	= \sum\limits_{\sigma_0,\ldots,\sigma_{m_--1}}
	\left[\prod\limits_{i=0}^{m_--1}
	W^{(n,1)}\left(\left.\begin{array}{cc}
		c_i & d_i\\
		c_{i+1} & d_{i+1}
	\end{array}\right|u-m_+-i\right)
	\right]
	\psi^{(n)}(z|0)^{c+\mu+\nu+\Sigma}_{c-m_+}, \,
	\end{align}
	where $\sigma_j = \pm 1, \; c_i = c+\mu+\nu+\sum_{\xi=0}^{i-1}\sigma_\xi, \; d_i = c+\mu-i, \; \Sigma = \sum_{\xi=0}^{m_--1}\sigma_\xi .$
	
	Let us fix $\Sigma$ and find the coefficient of the corresponding intertwining vector $\psi^{(n)}(z|0)^{c+\mu+\nu+\Sigma}_{c-m_+}$ in \eqref{second_product_action}. We will do the computation similar to that from the proof of proposition~\ref{vector_fusion_prop}. As can be seen, $c_{i+1}=c_i+\sigma_i$. Let $\varkappa_+$ be the number of values $+1$ among $\{\sigma_\xi\}_{\xi=0}^{m_--1}$, and $\varkappa_-$ be the number of values $(-1)$. Quantities $\varkappa_+$ and $\varkappa_-$ are uniquely determined by $\Sigma$:
	\begin{equation} \label{varkappa_pm}
	\varkappa_\pm(\Sigma) = \frac{1}{2}(m_-\pm\Sigma) \, .
	\end{equation}
	Denote $c^\prime_{i}\equiv\sum_{\xi=0}^{i-1}\sigma_\xi,\; i=0,1,\ldots,m_--1.$
	Let $\sigma_{i_p}=+1$ for $p=1,2,\ldots,\varkappa_+$. Then $c^\prime_{i_p}=-i_p+2(p-1),$
	and from \eqref{W_n_1_c_plus_1} it follows that
	\begin{align} 
	\nonumber &
	W^{(n,1)}\left(\left.\begin{array}{cc}
		c_{i_p} & d_{i_p}\\
		c_{i_p+1} & d_{i_p+1}
	\end{array}\right|u-m_+-i_p\right) \\
	\label{W_n_1_c_i_p} &
	= \frac{\left(\frac{1}{2}(n-\nu)-p+1\right)\left(-u+m_-+c-\frac{1}{2}(n-\nu)+p-1+w\right)}{c+\mu+\nu+c^\prime_{i_p}+w} \, .
	\end{align}
	Let $\sigma_{j_q}=-1$ for $q=1,2,\ldots,\varkappa_-$. Then $c^\prime_{j_q}=j_q-2(q-1),$
	and from \eqref{W_n_1_c_plus_1} one obtains that
	\begin{align}
	\nonumber &
	W^{(n,1)}\left(\left.\begin{array}{cc}
		c_{j_q} & d_{j_q}\\
		c_{j_q+1} & d_{j_q+1}
	\end{array}\right|u-m_+-j_q\right) \\
	\label{W_n_1_c_j_q} &
	=\frac{\left(u-m_++\frac{1}{2}(n+\nu)-q+1\right)\left(c+\mu+\frac{1}{2}(n+\nu)-q+1+w\right)}{c+\mu+\nu+c^\prime_{j_q}+w} \, .
	\end{align}
	Substituting \eqref{W_n_1_c_i_p} and \eqref{W_n_1_c_j_q} into \eqref{second_product_action} we find that the coefficient of $\psi^{(n)}(z|0)^{c+\mu+\nu+\Sigma}_{c-m_+}$ is equal to
	\begin{align}
	\label{coefficient_psi_Sigma} &
	(\theta_1)^-_{\varkappa_+}(\theta_2)^+_{\varkappa_+} (\theta_3)^-_{\varkappa_-} (\theta_4)^-_{\varkappa_-} \, \frac{x+\Sigma}{x}\sum\limits_{c^\prime_1,\ldots,c^\prime_{m_-}}\frac{1}{x+c^\prime_1}\times\ldots\times\frac{1}{x+c^\prime_{m_--1}}\times\frac{1}{x+c^\prime_{m_-}},
	\end{align}
	where
	\begin{align}
		\label{factorial} & (y)^\sigma_k \equiv \prod\limits_{j = 0}^{k-1} (y + \sigma j), \qquad \sigma = \pm, \\
		\label{theta12} & \theta_1 = \frac{1}{2}(n-\nu), \; \theta_2 = -u+c+m_--\frac{1}{2}(n-\nu)+w, \\
		\label{theta34} & \theta_3 = u-m_++\frac{1}{2}(n+\nu), \; \theta_4 = c+\mu+\frac{1}{2}(n+\nu)+w, \; x = c+\mu+\nu+w
	\end{align}
	and the summation is over all $c^\prime_1, \ldots, c^\prime_{m_-}$ such that
	\begin{equation} \label{c_prime_condition}
	|c^\prime_i-c^\prime_{i+1}|=1,\: i=0,1,\ldots,m_--1,\quad c^\prime_0 \equiv 0,\quad c^\prime_{m_-} = \Sigma.
	\end{equation}
	Denote
	\begin{equation} \label{f_varkappa_x}
	f(\varkappa_+,\varkappa_-|x) \equiv
	\sum\limits_{c^\prime_1,\ldots,c^\prime_{m_-}}\frac{1}{x+c^\prime_1}
	\cdot \ldots \cdot 
	\frac{1}{x+c^\prime_{m_--1}} \cdot \frac{1}{x+c^\prime_{m_-}} \, ,
	\end{equation}
	this is a function of integer arguments $\varkappa_+$, $\varkappa_-$ and the parameter $x$ (because by the virtue of \eqref{varkappa_pm} the numbers $m_-$ and $\Sigma$ are expressed in terms of $\varkappa_+$ and $\varkappa_-$). With the help of \eqref{c_prime_condition} and \eqref{f_varkappa_x} it is possible to obtain the recurrence relation for $f(\varkappa_+,\varkappa_-|x)$ in the case when $\varkappa_+ + \varkappa_- \geq 2$:
	\begin{equation} \label{f_varkappa_x_recurrence}
	f(\varkappa_+,\varkappa_-|x)=\frac{1}{x+\varkappa_+-\varkappa_-}[f(\varkappa_+-1,\varkappa_-|x)+f(\varkappa_+,\varkappa_--1|x)].
	\end{equation}
	For $\varkappa_++\varkappa_-=1$ the formula \eqref{f_varkappa_x} gives
	\begin{equation} \label{f_varkappa_x_base}
	f(1,0|x)=\frac{1}{x+1},\quad f(0,1|x)=\frac{1}{x-1}.
	\end{equation}
	Then using \eqref{f_varkappa_x_recurrence} we obtain $f(\varkappa_+,\varkappa_-|x)$ by means of induction with respect to $\varkappa_++\varkappa_-$ and with the base case \eqref{f_varkappa_x_base}:
	\begin{equation} \label{f_varkappa_x_explicit}
	f(\varkappa_+,\varkappa_-|x)=\frac{{\varkappa_++\varkappa_-\choose \varkappa_+}}{\prod\limits_{i=1}^{\varkappa_+}(x+i)\prod\limits_{j=1}^{\varkappa_-}(x-j)},
	\end{equation}
	where ${\varkappa_++\varkappa_-\choose \varkappa_+}$ is a binomial coefficient.
	
	Thus, substituting $f(\varkappa_+,\varkappa_-|x)$ from \eqref{f_varkappa_x_explicit} into \eqref{coefficient_psi_Sigma}, and then substituting \eqref{coefficient_psi_Sigma} into \eqref{second_product_action}, and \eqref{second_product_action} into \eqref{defining_relation_W_n_m}, one obtains that
	\begin{align} 
	\nonumber &
	(-\alpha)^{-m} \, P_1 \, P_2 \, \psi^{(n)}(z|0)^{c+\mu+\nu}_{c+\mu} \\
	\label{defining_relation_W_n_m_1} &
	= \sum\limits_{\Sigma}
	\frac{(x+\Sigma) \, {m_-\choose \varkappa_+} \, (\theta_1)^-_{\varkappa_+}(\theta_2)^+_{\varkappa_+} (\theta_3)^-_{\varkappa_-} (\theta_4)^-_{\varkappa_-} }{x \, (x+1)^+_{\varkappa_+}(x-1)^-_{\varkappa_-}} \, (-\alpha)^{-m_+} P_1 \, \psi^{(n)}(z|0)^{c+\mu+\nu+\Sigma}_{c-m_+} \, ,
	\end{align}
	where $\varkappa_+$ and $\varkappa_-$ are expressed in terms of $\Sigma$ with the help of \eqref{varkappa_pm}, and the sum over $\Sigma$ goes from $-m_-$ to $m_-$ with step $2$.
	
	Applying \eqref{defining_relation_W_n_1_minus} $m_+$ we find that
	\begin{align}
	\nonumber &
	(-\alpha)^{-m_+} P_1 \, \psi^{(n)}(z|0)^{c+\mu+\nu+\Sigma}_{c-m_+} \\
	\label{first_product_action} &
	= \sum\limits_{\tau_0,\ldots,\tau_{m_+-1}}\left[\prod\limits_{i=0}^{m_+-1}
	W^{(n,1)}\left(\left.\begin{array}{cc}
	\tilde{c}_i & \tilde{d}_i\\
	\tilde{c}_{i+1} & \tilde{d}_{i+1}
	\end{array}\right|u-i\right)\right]\psi^{(n)}(z|0)^{c+\mu+\nu+\Sigma+T}_{c},
	\end{align}
	where $\tau_j = \pm 1, \; \tilde{c}_i = c+\mu+\nu+\Sigma+\sum_{\xi=0}^{i-1}\tau_\xi, \; \tilde{d}_i = c-m_++i, \; T = \sum_{\xi=0}^{m_+-1}\tau_\xi.$ Similarly to \eqref{second_product_action}, \eqref{first_product_action} can be transformed into
	\begin{align} 
	\nonumber & (-\alpha)^{-m_+} P_1 \, \psi^{(n)}(z|0)^{c+\mu+\nu+\Sigma}_{c-m_+} \\
	\label{first_product_action_1} & = \sum\limits_{T}\frac{(x+\Sigma+T) \, {m_+\choose \rho_+(T)} \, (\theta_5)^-_{\rho_+(T)} (\theta_6)^+_{\rho_+(T)} (\theta_7)^-_{\rho_-(T)} (\theta_8)^-_{\rho_-(T)}}{(x+\Sigma) \, (x+\Sigma+1)^+_{\rho_+(T)} (x+\Sigma-1)^-_{\rho_-(T)}}
	\, \psi^{(n)}(z)^{c+\mu+\nu+\Sigma+T}_{c}
	\end{align}
	where
	\begin{align}
	\label{rho_pm} & \rho_\pm(T) = \frac{1}{2}(m_+\pm T), \\
	\label{theta56} & \theta_5 = u + \frac{1}{2}(n-\nu-\Sigma-m_-), \quad
	\theta_6 = c + \mu - \frac{1}{2}(n-\nu-\Sigma+m_-) + w, \\
	\label{theta78} & \theta_7 = \frac{1}{2}(n+\nu+\Sigma+m_-), \quad
	\theta_8 = u + c + \mu + \frac{1}{2}(n+\nu+\Sigma-m_-) + w \, ,
	\end{align}
	and the sum over $T$ goes from $-m_+$ to $m_+$ with step $2$.
	
	Substituting \eqref{first_product_action_1} into \eqref{defining_relation_W_n_m_1} and comparing the obtained expression with the formula \eqref{defining_relation_W_n_m} for calculation of Boltzmann weights we conclude that
	\begin{align} 
		\nonumber &
		W^{(n,m)}\left(\left.\begin{array}{cc}
			c+\mu+\nu & c+\mu\\
			c+\mu+\nu+\mu^\prime & c
		\end{array}\right|u\right) \\
		\nonumber &
		= \sum\limits_{\Sigma}
		\frac{(x+\mu') \, {m_-\choose \varkappa_+} \, {m_+\choose \rho_+(\mu'-\Sigma)} \, (\theta_1)^-_{\varkappa_+}(\theta_2)^+_{\varkappa_+} (\theta_3)^-_{\varkappa_-} (\theta_4)^-_{\varkappa_-} }{ x \, (x+1)^+_{\varkappa_+} (x-1)^-_{\varkappa_-} \, (x+\Sigma+1)^+_{\rho_+(\mu'-\Sigma)} (x+\Sigma-1)^-_{\rho_-(\mu'-\Sigma)} } \\
		\label{W_n_m} & \times  (\theta_5)^-_{\rho_+(\mu'-\Sigma)} (\theta_6)^+_{\rho_+(\mu'-\Sigma)} (\theta_7)^-_{\rho_-(\mu'-\Sigma)} (\theta_8)^-_{\rho_-(\mu'-\Sigma)} \, ,
	\end{align}
	where the sum over $\Sigma$ goes from $-m_-$ to $m_-$ with step $2$, the numbers $\varkappa_\pm$ and $\rho_\pm$ are defined in \eqref{varkappa_pm} and \eqref{rho_pm}, the parameters $x$ and $\theta_i$ are defined in \eqref{theta12}, \eqref{theta34}, \eqref{theta56}, \eqref{theta78}, and the notation $(y)^\pm_k$ was introduced in \eqref{factorial}.
	
	By means of \eqref{W_n_m} the Boltzmann weight
	$W^{(n,m)}\left(\left.\begin{array}{cc}
		a & b \\
		b' & c
	\end{array}\right|u\right)$
	can be expressed in the form of terminating hypergeometric series
	\begin{equation} \nonumber
	{}_pF_q\left[\left.\begin{array}{ccc}
	\alpha_1 & \ldots & \alpha_p\\
	\beta_1 & \ldots & \beta_q
	\end{array}\right|z\right]\equiv\sum\limits_{k=0}^{\infty}\frac{(\alpha_1)_k\ldots (\alpha_p)_k}{(\beta_1)_k\ldots (\beta_q)_k}\frac{z^k}{k!},
	\end{equation}
	where $(y)_k = \prod_{j=0}^{k-1}(y+j) .$ Denote
	\begin{equation} \nonumber
		n_\pm=\frac{1}{2}[n\pm(b-a)], \quad m_\pm=\frac{1}{2}[m\pm(c-b)], \quad m^\prime_\pm=\frac{1}{2}[m\pm(b^\prime-a)] \, .
	\end{equation}
	The expression takes the form
	\begin{align}
		\nonumber &
		W^{(n,m)}\left(\left.\begin{array}{cc}
			a & b\\
			b^\prime & c
		\end{array}\right|u\right)
		\\
		\nonumber &
		= C^{(n,m)}(a,b,c,b^\prime|u) \; 
		{}_{9}F_{8}\left[\left.\begin{array}{cccccccccc}
			\alpha_1 & \alpha_2 & \alpha_3 & \alpha_4 & \alpha_5 & \alpha_6 & \alpha_7 & \alpha_8 & \alpha_9 \\
			\beta_1 & \beta_2 & \beta_3 & \beta_4 & \beta_5 & \beta_6 & \beta_7 & \beta_8 
		\end{array}\right|1\right] ,
	\end{align}
	where in the case $b + b' \leq a + c$:\\
	\begin{tabular}{|c|c|c|}\hline
		$j$ & $\alpha_j$ & $\beta_j$ \\\hline
		1 & $-m_-$ & $a+w+1$ \\\hline
		2 & $-n_+$ & $u-m + n_- + 1$ \\\hline
		3 & $-m_+'$ & $c + n_- - m_+ + 1 + w$ \\\hline
		4 & $a - m_- + w$ & $1 - \nicefrac{1}{2}(b'+b-a-c)$ \\\hline
		5 & $-u + c - n_+ + m_- + w$ & $\nicefrac{1}{2}(b'-b+a+c)+w+1$ \\\hline
		6 & $\nicefrac{1}{2}(a - m_- + w+2)$ & $\nicefrac{1}{2}(a - m_- + w)$ \\\hline
		7 & $a - m'_- + w$ & $-u-n_+$ \\\hline
		8 & $n_- + 1$ & $c- m_+ - n_+ + w$ \\\hline
		9 & $u+c - m_+ + n_- + w+1$ & \\\hline
	\end{tabular}
	\begin{align}
		\nonumber & C^{(n,m)}(a,b,c,b^\prime|u)
		= (b'+w) \, {m_+ \choose m'_+} \,
		\frac{\Gamma(a-m'_-+w)}{\Gamma(a+w+1)}
			\left(n_-+\nicefrac{1}{2}(b+b'-a-c)+1\right)_{-\nicefrac{1}{2}(b+b'-a-c)}
		\\
		\nonumber & 
		\times
		\frac{\Gamma\left(b+n_-+1+w\right)}{\Gamma\left(c+n_--m_++1+w\right)}
		\frac{\Gamma(a-m_-+w+1)}{\Gamma\left(\nicefrac{1}{2}(b'-b+a+c)+w+1\right)}
		\frac{\Gamma(u+c-m_++n_-+w+1)}{\Gamma(u+b+n_--m'_-+w+1)}
		\\
		\nonumber &
		\times
		\frac{\Gamma\left(u+n_++1\right)}{\Gamma(u+n_+-m'_++1)}
		\frac{\Gamma\left(\nicefrac{1}{2}(b+b'+c-a)-n_++w\right)}{\Gamma(c-m_+-n_++w)}\frac{\Gamma(u-n_--m_++1)}{\Gamma(u-m+n_-+1)}
	 \, ,
	\end{align}
	and in the case $b + b' \geq a + c$:\\
	\begin{tabular}{|c|c|c|}\hline
		$j$ & $\alpha_j$ & $\beta_j$ \\\hline
		1 & $-m'_-$ & $\nicefrac{1}{2}(b+b'+a-c)+w+1$ \\\hline
		2 & $-n_+ + \nicefrac{1}{2}(b+b'-a-c)$ & $u-m_+ -m'_- + n_- + 1$ \\\hline
		3 & $-m_+$ & $b + n_- - m'_- + 1 + w$ \\\hline
		4 & $a - m'_- + w$ & $1$ \\\hline
		5 & $-u + b - n_+ + m'_+ + w$ & $b'+w+1$ \\\hline
		6 & $\nicefrac{1}{2}(b' - m_+ + w+2)$ & $\nicefrac{1}{2}(b' - m_+ + w)$ \\\hline
		7 & $b' - m_+ + w$ & $-u-n_+ + \nicefrac{1}{2}(b+b'-a-c)$ \\\hline
		8 & $n_- + 1 + \nicefrac{1}{2}(b+b'-a-c)$ & $b- m'_- - n_+ + w$ \\\hline
		9 & $u+b - m'_- + n_- + w+1$ & \\\hline
	\end{tabular}
	\begin{align}
		\nonumber & C^{(n,m)}(a,b,c,b^\prime|u)
		= {m_- \choose m'_-} \,
		\frac{\Gamma(a - m'_- + w)}{\Gamma\left(\nicefrac{1}{2}(b+b'+a-c)+w+1\right)}
		\left(-\nicefrac{1}{2}(n-\nu)\right)_{\nicefrac{1}{2}(b+b'-a-c)} \\
		\nonumber &
		\times
		\frac{\Gamma(-u+b - n_+ + m'_+ + w)}{\Gamma(-u+c - n_+ + m_- + w)}
		\frac{\Gamma(u- n_- + m_+ + 1)}{\Gamma(u- m_+ - m'_- + n_- + 1)}
		\frac{\Gamma(b + n_- + 1 + w)}{\Gamma(b + n_- - m'_- + 1 + w)}
		\\
		\nonumber &
		\times
		\frac{\Gamma(b' - m_+ + w + 1)}{\Gamma(b'+w)}
		\frac{\Gamma\left(u + n_+ - \nicefrac{1}{2}(b+b'-a-c) + 1\right)}{\Gamma(u + n_+ - m'_+ + 1)}
		\frac{\Gamma\left(\nicefrac{1}{2}(b+b'-a+c) - n_+ + w\right)}{\Gamma(b - m'_- - n_+ + w)} \, .
	\end{align}
	In accordance with propositions \ref{YBE_SOS_from_vertex_prop} and \ref{vectors_linear_independence_prop}, $w \notin \mathbb{Z}$.

	\section{Connection between the $7$-vertex and the $11$-vertex models} \label{11V_section}
	
	Consider an arbitrary family of $R$-operators $R^{(n,m)}$ satisfying the Yang-Baxter equation \eqref{YBE_vert}. Consider the following similarity transformation:
	\begin{equation} \label{conj}
		\widetilde{R}^{(n,m)} = [A^{(n)}(u) \otimes A^{(m)}(v)] \, R^{(n,m)}(u-v) \, [A^{(n)}(u)^{-1} \otimes A^{(m)}(v)^{-1}] \, ,
	\end{equation}
	where $A^{(n)}(u)$ and $A^{(m)}(v)$ are invertible operators. If the obtained operator $\widetilde{R}^{(n,m)}$ depends only on the difference $u-v$,
	then any triple of such operators also satisfy the Yang-Baxter equation \cite{zabrodin}:
	\begin{equation} \nonumber
		\widetilde{R}^{(k,n)}_{12}(v) \, \widetilde{R}^{(k,l)}_{13}(u) \, \widetilde{R}^{(n,l)}_{23}(u-v) 
		= \widetilde{R}^{(n,l)}_{23}(u-v) \, \widetilde{R}^{(k,l)}_{13}(u) \, \widetilde{R}^{(k,n)}_{12}(v) \, .
	\end{equation}
	
	Let us construct the similarity transformation \eqref{conj} for $R$-operators obtained from the seven-vertex model by means of fusion procedure. Consider the transform \eqref{conj} for $R$-matrix \eqref{R_7V_rat} with $A^{(1)}(u)$ of the following form:
	\begin{equation} \nonumber
		A^{(1)}(u) = 
		\begin{pmatrix}
			1 & 0\\
			-\alpha u & 1
		\end{pmatrix} .
	\end{equation}
	As a result we find
	\begin{equation} \nonumber
		\widetilde{R}^{(1,1)}(u-v) =
		\begin{pmatrix}
			u-v+1 & 0 & 0 & 0\\
			\alpha(u-v) & (u-v) & 1 & 0\\
			-\alpha(u-v) & 1 & u-v & 0\\
			\alpha^2(u-v) & \alpha(u-v) & -\alpha(u-v) & u-v+1
		\end{pmatrix} .
	\end{equation}
	This is the $R$-matrix of $11$-vertex model \cite{AZ_23,DKK_03,KS_97}.
	
	The operator $R^{(n,m)}$ \eqref{R_k_n}, obtained from the operator $R^{(1,1)}$ of the seven-vertex model with the help of fusion procedure acts in $S^n\mathbb{C}^2 \otimes S^m\mathbb{C}^2$. The operator $A^{(n)}(u)$ is constructed from $A^{(1)}(u)$ in the following way:
	\begin{equation} \nonumber
		A^{(n)}(u) = \left. [A^{(1)}(u)]^{\otimes n} \right|_{S^n\mathbb{C}^2} \, .
	\end{equation}
	Using \eqref{R_k_1}, \eqref{R_k_n} and the identities
	\begin{equation} \nonumber
		A^{(1)}(u) \, A^{(1)}(v) = A^{(1)}(v) \, A^{(1)}(u) = A^{(1)}(u + v), \qquad [A^{(1)}(u)]^{\otimes n} \, \Pi_{1\ldots n} = \Pi_{1\ldots n} \, [A^{(1)}(u)]^{\otimes n} \, 
	\end{equation}
	one can prove that $\widetilde{R}^{(n,m)}$ indeed depends only on the difference $u-v$.

	In that way, we automatically obtain the fusion procedure for the $11$-vertex model from the fusion procedure for the $7$-vertex model: we have found the family of operators $\widetilde{R}^{(n,m)}(u)$ satisfying the Yang-Baxter equation \eqref{YBE_vert}, and the simplest representative of this family is the $R$-operator of the $11$-vertex model $\widetilde{R}^{(1,1)}(u)$.
	
	In the realization on the space of polynomials of one variable $z$ (see section~\ref{polynomial_sect}) the operator $A^{(n)}(u)$ shifts the argument of a function by $\alpha u$:
	\begin{equation} \label{An}
		A^{(n)}(u) = e^{\alpha u \partial_z} \, .
	\end{equation}
	
	Taking into account \eqref{conj}, the vertex-SOS correspondence for $\widetilde{R}^{(n,m)}$ assumes the form
	\begin{equation} \nonumber
		\widetilde{R}^{(n,m)}(u-v) \, \Psi^{(n)}{}^a_b\otimes\Psi^{(m)}{}^b_c = \sum\limits_{b^\prime} \Psi^{(n)}{}^{b^\prime}_c\otimes\Psi^{(m)}{}^a_{b^\prime} \, W^{(n,m)}\left(\left.\begin{array}{cc}
			a & b\\
			b^\prime & c
		\end{array}\right|u-v\right),
	\end{equation}
	where $W^{(n,m)}$ are the same SOS Boltzmann weights that were for the fused $7$-vertex SOS model, and the intertwining vectors
	\begin{equation} \nonumber
		\Psi^{(n)}{}^a_b = A^{(n)}(u) \, \psi^{(n)}(u)^a_b
	\end{equation}
	do not depend on the spectral parameter and in the realization on polynomials of one variable have the form
	\begin{equation} \label{intertwining_vector_11V}
		\Psi^{(n)}{}^a_b(z) = (-1)^n\prod_{p=1}^{n_+}[z-\alpha(n-a-2p+1-t)]
		\prod_{q=1}^{n_-}[z-\alpha(n+a-2q+1+s)] ,
	\end{equation}
	where $n_\pm = \nicefrac{1}{2}(n \pm (b-a))$. The formula \eqref{intertwining_vector_11V} follows from \eqref{intertwining_vector_polynomial_z} and \eqref{An}.

	\section{Conclusion}
	
	We have calculated the Boltzmann weights $W^{(n,m)}$ of SOS models in the form of terminating hypergeometric series ${}_{9}F_8$, and in a particular case -- in terms of rational fractions. The realization of the $R$-operator and the intertwining vectors on the space of polynomials has notably simplified the computations. Statistical weights of SOS models depend on additional real-valued free parameter $w \notin \mathbb{Z}$ and do not depend on the parameter $\alpha$ of the seven-vertex model \eqref{R_7V_rat}. For $n = m = 1$ the connection with the SOS model from the paper \cite{zapiski} was found.
	
	The obtained intertwining vectors can be used for derivation of the explicit expression connecting the partition functions of vertex and SOS models, as well as for construction of vertex model transfer-matrix eigenvectors~\cite{vega}.
	
	In addition, with the help of the similarity transformation from the $R$-operators of the fused $7$-vertex model we constructed the new family of Yang-Baxter equation solutions, the simplest representative of this family is the $R$-matrix of the $11$-vertex model \cite{AZ_23,DKK_03,KS_97}. The fusion procedure for the $11$-vertex model was automatically obtained from the fusion procedure for the $7$-vertex model. Using the same similarity transformation we obtained the vertex-SOS correspondence for the new family of $R$-matrices, including the $11$-vertex model, from the vertex-SOS correspondence for the $7$-vertex model. In so doing, the SOS Boltzmann weights $W^{(n,m)}$ turn out to be the same as for the $7$-vertex model, and the intertwining vectors become independent of the spectral parameter, which can noticeably simplify the use of the vertex-SOS transform.
	
	The next step in the study of the obtained SOS models can be the computation of correlation functions \cite{bogolubov,konno}.

\end{document}